\documentclass[a4paper,onecolumn,11pt,unpublished]{quantumarticle}
\pdfoutput=1

\usepackage[utf8]{inputenc}
\usepackage[english]{babel}
\usepackage[T1]{fontenc}
\usepackage{amssymb}
\usepackage{amsmath}
\DeclareMathOperator*{\argmin}{arg\,min}
\usepackage{hyperref}
\usepackage{braket}
\usepackage{float}
\usepackage[dvipsnames]{xcolor}
\usepackage{graphicx}
\usepackage{algorithm}
\usepackage{algpseudocode}
\usepackage{caption}
\usepackage{subcaption}
\usepackage{lipsum}
\usepackage{amsthm}
\usepackage{tikz}
\usetikzlibrary{quantikz}
\usepackage[normalem]{ulem}
\usepackage{nomencl}
\newtheorem{theorem}{Theorem}

\usepackage[numbers,sort&compress]{natbib}
\begin{document}

\title{Quantum Approximation Optimization Algorithm for the Trellis based Viterbi Decoding of Classical Error Correcting Codes}
\author{Mainak Bhattacharyya}
\email{mainak23@iiserb.ac.in}
\orcid{https://orcid.org/0009-0001-5761-1289}
\affiliation{Department of Electrical Engineering and Computer Science, Indian Institute of Science Education and Research Bhopal, India}

\author{Ankur Raina}
\email{ankur@iiserb.ac.in}
\orcid{https://orcid.org/0000-0002-9022-3595}
\affiliation{Department of Electrical Engineering and Computer Science, Indian Institute of Science Education and Research Bhopal, India}
\maketitle
\begin{abstract}
We construct a hybrid quantum-classical Viterbi decoder for the classical error-correcting codes.
Viterbi decoding is a trellis-based procedure for maximum likelihood decoding of classical error-correcting codes.
In this article, we demonstrate that the quantum approximate optimization algorithm can find any path on the trellis with the minimum Hamming distance relative to the received erroneous vector.
We construct a generalized method to map the Viterbi decoding problem into optimization of a parameterized quantum circuit for any classical linear block code.
Also, we propose a uniform parameter optimization strategy to optimize the parameterized quantum circuit using a classical optimizer.
We observe that the proposed method efficiently generates low-depth trainable parameterized quantum circuits. 
Our approach makes the hybrid decoder more efficient than previous attempts at making quantum Viterbi algorithm.
We show that using uniform parameter optimization, we obtain parameters more efficiently for the parameterized quantum circuit than previously used methods such as random sampling and fixing the parameters.
\end{abstract}
\section{Introduction}
Classical error correcting codes are designed to tackle the noise occurring in communication and data storage channels. 
Encoding in these codes corresponds to introducing redundancy into the information bits. 
The resultant bits form the codeword.
Linear block codes, convolution codes are some widely known error correcting codes.
A good code consisting of codewords must have an efficient decoding technique alongside the encoding process to render the error correction efficient.
After receiving an erroneous signal/vector from the communication channel, the decoder tries to find the most likely transmitted codeword or the error pattern,
which is often considered as a hard problem.
A standard but classically exhaustive decoding process is the maximum likelihood (ML) decoding.
Given any received word, the ML decoding calculates the probabilities of all possible codewords.
The codeword having the maximum probability is decoded as the best candidate in the codespace (i.e. the set of codewords).
The searching process in the codespace using the maximum likelihood approach has a time complexity exponential to the number of bits $n$, i.e. $\mathcal{O}(2^n)$ and it is NP-hard for a general class of linear codes \cite{guruswami2005maximum, berlekamp}.
ML decoding can be cast as a form of syndrome decoding or the Viterbi decoding and it is proven to be NP-complete by Berlekamp \emph{et al}. in 1978 \cite{berlekamp}.
This suggests there exists no algorithm to solve the problem in polynomial time in the number of bits of the codeword.
In 1990, Bruck and Naor showed that in general, the knowledge of the code does not improve the decoding complexity \cite{bruck}.
The authors argued that if there exists an efficient decoder, one can process the code defined by it's parity check matrix to find an efficient decoder in finite ammount of processing.
As in general practice the code to be used might not be known in advance, the prior mentioned workflow is referred as the preprocessing.
Their sounding contribution was to show that a preprocessing for an unbounded time does not guarantee an efficient decoder.
This suggests that there are some linear block codes, which cannot be decoded efficiently.
A natural direction to look for is the quantum advantage in classical decoding problems.\\\\
During 1992-1997, quantum algorithms showed the promise to provide speedup over the corresponding classical counterparts for certain problems \cite{deutsch,bernstein_vazzirani,simon}. 
One of them introduced by LK Grover in 1996 \cite{grover}, is considered to be one of the prime quantum algorithms showing quadratic speed up compared to the classical algorithms of similar interest having worst case time complexity, which is exponential in the size of the input.
The worst case time complexity refers to the case, where an algorithm executes maximum number of operations on a fixed input size to fulfil the task.
This leads to an analysis on the upper bound of the running time of an algorithm.
Grover's algorithm uses oracle based approach to search for an item in a disordered data set with quadratic speedup. 
Jung \emph{et al.} used the inspiration from Grover's search to achieve a quadratic speed up over classical ML decoding \cite{qml}.
Quadratic speed up is also achieved for the NP-complete syndrome decoding problem using a modified Grover's search algorithm \cite{syndrome_us}. 
The authors argued that the Grover inspired algorithm significantly reduces the time complexity of decoding than the classically exponential one.
The problem of finding the path with minimum path metric, for example in Viterbi decoding can be addressed by a quantum minima finding algorithm given by D\"urr and H{\o}yer \cite{durr_hoyer}. 
This can bring a quadratic speed up in the search process.
In general this type of algorithms are inefficient due to their vulnerability to NISQ devices \cite{preskil}.
Therefore, we look for other available alternatives, which are most suitable for implementation in the NISQ devices \cite{bharti2022noisy,mcclean2016theory}.\\\\
Variational quantum algorithms are an efficient alternative for such cases \cite{cerezo2021variational,callison2022hybrid}.
The quantum approximation optimization algorithm (QAOA) primarily focuses to solve the combinatorial optimization problems.
It is used to find an approximate solution to problems like MAX-CUT problem \cite{qaoa}, Travelling Salesman Problem \cite{tsp} and more \cite{awasthi2023quantum, moussa2022unsupervised, streif2021beating, borle2021quantum, vikstaal2020applying}.
It is thought to be one of the best candidate for NISQ devices, due to its efficiency for circuits with lower depths \cite{farhi2016quantum}.
Quantum approximation optimization algorithm (QAOA) has previously been explored to study the quasi maximum likelihood (ML) decoding of the classical channel codes \cite{matsumine2019channel}, where the authors map the quasi ML decoding problem into ground state determination of an Ising Hamiltonian and used a simple level-$1$ QAOA to study the success rate of the QAOA based algorithm.
In this work we provide a QAOA based approach towards the trellis based decoding of classical codes.
We provide a general map of cost Hamiltonian corresponding to a cost function for the Viterbi decoding problem.
We also show that the depth of a general mixer unitary for the QAOA based Viterbi decoder, implicitly depends on the minimum Hamming distance $d$ and the number of bits in the codeword.
Therefore, the circuit depth of the PQC still depends on the structure of the classical code itself.
However, we argue that QAOA still has better efficiency than the specialized iterative amplitude amplification (IAA) based algorithm previously used to study a variant of quantum Viterbi algorithm for the decoding of classical convolution codes \cite{qva}.
In their work Grice and Meyer showed the IAA based algorithm followed by an encoding of path probabilities into the phases of the superposition states increases the probability of the decoded state.
This has a quadratic depth dependence on fanout and the number of bits used to represent the state.
This results in an overhead, which is inefficient for the current quantum computing hardware.
In Section \ref{sec:Performance of the Hybrid Decoder} we show that the low depth performance of QAOA based Viterbi decoder is very accurate using a uniform parameter optimization-based structured PQC.
Therefore, the hybrid QAOA-Viterbi decoder is indicated as a potentially efficient substitute for the previous quantum Viterbi decoder.
We thus, give a Quantum-classical hybrid approach for the Viterbi decoding, inspired from the QAOA introduced by Farhi \emph{et al}. in 2014 \cite{qaoa}.\\\\
The success of finding a good approximate solution using QAOA depends on finding a set of good parameters for the PQC.
The choice of good parameters depends on properly choosing the initial parameters and cost functions \cite{cerezo}, otherwise due to the presence of flat optimization landscape it is often very hard to find the global minima.
In their work, Bittel and Kliesch showed that the classical optimization of the parameterized quantum circuits is NP hard \cite{np_hard}. 
To address the issue, we suggest a simple uniform parameter optimization (UPO) technique to determine good parameters for the PQC circuit. 
We observe that the proposed heuristic approach towards optimization is better performant than both the random initialization method and fixed parameter optimization (FPO) strategy suggested by Lee \emph{et al}. \cite{fixed_parameter}.
Although a proper runtime complexity of the decoder is correlated with the optimization method used,
QAOA is expected to approach a problem with better runtime complexity using the PQC and the classical optimizer.
In spite of many obstacles, open studies are going on to check whether it is the general case or not.
Many reported instances show that QAOA might still produce a quantum advantage over the classical counterparts
\cite{crooks2018performance,blekos2023review,an2022quantum}.
In similar direction, Zhang \emph{et al.} demonstrated that QAOA provides exponential acceleration for the Minimum Vertex Cover (MVC) problem over the classical competitive decision algorithm \cite{zhang2022applying}.
The unstructured search problem presented using QAOA by Jiang \emph{et. al.} provides a Grover-like quadratic speed-up \cite{jiang2017near}.
Recently J. Golden \emph{et al.} showed numerical evidence of exponential speed-up for QAOA over Grover type unstructured search in constrained optimization problems \cite{golden2022evidence}.
We explore all the above scenarios and show improvement heuristically in our algorithm.
\\\\
The paper is organized as follows.
We briefly discuss the trellis-based Viterbi decoding for classical linear block codes in Section \ref{sec:Viterbi decoding of classical codes}. 
We present the workflow of QAOA in Section \ref{sec:qaoa}.
We formally introduce the quantum-classical (hybrid) Viterbi decoder by encoding the Viterbi decoding problem for the QAOA in Section \ref{sec:Quantum Classical Hybrid Viterbi Decoder}. 
We show a general method for constructing the mixer Hamiltonian depends on the minimum Hamming distance of the code in Theorem \ref{theorem_mixer}.
We propose a uniform parameter optimization (UPO) strategy in Section \ref{sec:UPO} and analyze the performance of decoding for the linear block codes.
We perform all the simulations of the algorithm using $qiskit$ and $qasm$ simulator of IBM Quantum Experience.
\makenomenclature
\renewcommand{\nomname}{List of notations}
\renewcommand{\nompreamble}{Here, we describe certain symbols we use frequently throughout the paper.}
\nomenclature{\(\mathbf{G}\)}{Generator matrix of a linear block code.}
\nomenclature{\(\mathbf{H}\)}{Parity check matrix of a linear block code.}
\nomenclature{\(f(x)\)}{Cost function dependent on codeword $x$.}
\nomenclature{\(H_f\)}{Cost function Hamiltonian.}
\nomenclature{\(U_f\)}{Cost unitary}
\nomenclature{\(H_m\)}{Mixer Hamiltonian.}
\nomenclature{\(U_m\)}{Mixer/ Mixer unitary}
\nomenclature{\(\ket{\psi_{\mathrm{in}}}\)}{Initial guess of the solution state, and must be a member of the codespace.}
\nomenclature{\(c\)}{Codeword belonging to $\mathcal{C}$.}
\nomenclature{\(y\)}{received erroneous vector.}
\nomenclature{\(w_H(c)\)}{Hamming weight of the codeword.}
\nomenclature{\(d\)}{minimum Hamming distance of the code.}
\nomenclature{\(d(c,y)\)}{Hamming distance between two vectors $c$ and $y$.}
\nomenclature{\(c_j\)}{$j^{th}$ bit of a codeword.}
\nomenclature{\(\mathcal{C}\)}{Codespace}
\nomenclature{\(U_g\)}{Unitary for creating uniform superposition of all valid trellis paths.}
\nomenclature{\(f\)}{Array for storing the expectation of cost function value at each sampling.}
\nomenclature{\(t_\beta, t_\gamma\)}{Optimized parameters for each sampling instances.}
\nomenclature{\(q\)}{Number of samples for opting out best parameters.}
\nomenclature{\(p_1, p_2\)}{Optimized Parameters.}
\nomenclature{\(p\)}{Number of unitary layers applied in the parameterized quantum circuit in QAOA.}
\nomenclature{\(\alpha\)}{Approximation ratio with respect to the minimum value of cost function.}
\nomenclature{\(\mathrm{X}\)}{Pauli-$\mathrm{X}$ gate.}
\nomenclature{\(\mathrm{Z}\)}{Pauli-$\mathrm{Z}$ gate.}
\nomenclature{\(\mathrm{H}\)}{$\mathrm{Hadamard}$ gate.}
\nomenclature{\(0_{1 \times n}\)}{Row vector of $n$ zeros.}
\nomenclature{\(n\)}{Number of bits in the codeword.}
\nomenclature{\(k\)}{Number of bits in the information/ message.}
\nomenclature{\(m_i\)}{Number of occurrence of $i^{th}$ state upon measuring the PQC.}
\printnomenclature
\section{Viterbi decoding of classical codes}
\label{sec:Viterbi decoding of classical codes}
Classical error correcting codes are designed to encode a message $m \in \{0,1\}^k$ into a codeword $c \in \{0,1\}^n$ based on a generator matrix $\mathbf{G}$.
For example, linear block codes encode $k$ bits of message $m$ into $n$ bits of codeword $c$ according to $c = m\mathbf{G}$.
The device which physically converts a set of input message bits, to the codewords of the error-correcting code is called the encoder.
The codewords of linear block code have systematic form and represent information bits followed by redundant bits. 
Thus, the codewords form a block-like structure.
The classical codes are denoted by $[n,k,d]$, where $d$ is the minimum Hamming distance of the code.
The minimum number of bits in which any two codeword in the codespace differs is denoted by $d$.
Further the parameter $d$ determines the error correcting capability of the code.
For instance a classical code with Hamming distance $d$ can correct errors upto Hamming distance $\frac{d-1}{2}$ with respect to the codewords.
The decoding problem of these classical codes are extensively based on the parity check matrix $\mathbf{H}$ of the code, which has a dimension of $n-k \times n$.
The members of the codespace $c \in \mathcal{C}$ satisfy $c\mathbf{H}^{T} = 0_{1 \times {(n-k)}}$, where the codeword $c$ is a row vector of dimensions $1 \times n$.
The codespace is thus defined as
\begin{align}
    \mathcal{C} = \{c: c\mathbf{H}^{T} = 0_{1 \times {(n-k)}}\}.
\end{align}
The cardinality of the codespace $\mathcal{C}$ is $2^k$.
During the transmission, channel errors $\eta \in \{0,1\}^n$ may affect the codeword, such that the codeword $c$ turns into an erroneous vector $y \in \{0,1\}^n$:
\begin{align}
    y = c + \eta.
\end{align}
A maximum likelihood (ML) decoder for the code tries to find the most probable error pattern occurred during the transmission, upon receiving the erroneous vector $y$.
Syndrome decoding is one of the ML decoding methods that takes syndrome $s \in \{0,1\}^{n-k}$ as an input and determines the minimum Hamming weight solution to the following equation:
\begin{align}
    \label{eq:ss}
    s = y\mathbf{H}^T = \eta \mathbf{H}^T,
\end{align}
where, the parity check matrix is known a-priori.
In simple terms a syndrome decoder obtains a minimum weight $\eta$ satisfying Eq. \ref{eq:ss}.\\\\
Viterbi decoding is another ML decoder for decoding the codeword $c$ after receiving the erroneous vector $y$.
For the decoding of convolution codes Viterbi algorithm, introduced by Viterbi is based on the maximum a posteriori probability calculations \cite{viterbi}.
The general problem is given a erroneous vector, the decoder determines the codeword having the maximum a posteriori probability.
A trellis based approach of Viterbi decoding is thoroughly discussed by Forney \cite{forney_viterbi}.
For each time instant, the trellis contains nodes describing the state of the encoder and from each node there are branches.
The branches depend on the current input to the encoder.
Any path traversing through the trellis is a valid codeword. 
Given the received vector, the shortest trellis path (i.e. with the maximum aposteriori probability) is determined as the decoded codeword.
The path metric refers to the total Hamming distance between a path of the trellis and the received erroneous vector. 
For convolution codes, construction of the trellis is fairly simple and is discussed by Lin and Costello in their book \cite{lin}.
In a paper by Bahl. L. \emph{et al.} \cite{bcjr}, the trellis construction for linear block codes was introduced, using the parity check matrix $\mathbf{H}$ of the corresponding $[n,k,d]$ linear block code.
Any trellis based decoding might as well be applied to the linear block codes.
The state at time $t$ for a linear block encoder is given by \begin{align}
\label{eq:lbc_trellis}
    S_t = S_{t-1} + c_th_t,
\end{align}
where $t = 1,2,......,n$. $c_t$ is the $t^\mathrm{th}$ bit of the codeword and $h_t$ is the $t^\mathrm{th}$ column vector of parity check matrix $\mathbf{H}$. 
Each state $S_t$ is a vector of length $(n-k)$. 
Thus, the total number of all possible states are $2^{n-k}$. 
\begin{figure}[ht!]
\advance\leftskip-1.2cm
\begin{tikzpicture}[node distance={2.8 cm}, main/.style = {draw, circle,fill=cyan,thick}]

\node[main] (1) {$\mathbf{000}$};
\node[main] (2) [right of=1] {$\mathbf{000}$};
\node[main] (3) [below of=2] {$\mathbf{001}$};
\node[main] (4) [below of=3] {$\mathbf{010}$};
\node[main] (5) [below of=4] {$\mathbf{011}$};
\node[main] (6) [below of=5] {$\mathbf{100}$};
\node[main] (7) [below of=6] {$\mathbf{101}$};
\node[main] (8) [below of =7] {$\mathbf{110}$};
\node[main] (9) [below of=8] {$\mathbf{111}$};
\node[main] (10)[right of =2] {$\mathbf{000}$};
\node[main] (11) [below of=10] {$\mathbf{001}$};
\node[main] (12) [below of=11] {$\mathbf{010}$};
\node[main] (13) [below of=12] {$\mathbf{011}$};
\node[main] (14) [below of=13] {$\mathbf{100}$};
\node[main] (15) [below of=14] {$\mathbf{101}$};
\node[main] (16) [below of =15] {$\mathbf{110}$};
\node[main] (17) [below of =16] {$\mathbf{111}$};
\node[main] (18)[right of =10] {$\mathbf{000}$};
\node[main] (19) [below of=18] {$\mathbf{001}$};
\node[main] (20) [below of=19] {$\mathbf{010}$};
\node[main] (21) [below of=20] {$\mathbf{011}$};
\node[main] (22) [below of=21] {$\mathbf{100}$};
\node[main] (23) [below of=22] {$\mathbf{101}$};
\node[main] (24) [below of =23] {$\mathbf{110}$};
\node[main] (25) [below of =24] {$\mathbf{111}$};
\node[main] (26)[right of =18] {$\mathbf{000}$};
\node[main] (27) [below of=26] {$\mathbf{001}$};
\node[main] (28) [below of=27] {$\mathbf{010}$};
\node[main] (29) [below of=28] {$\mathbf{011}$};
\node[main] (30) [below of=29] {$\mathbf{100}$};
\node[main] (31) [below of=30] {$\mathbf{101}$};
\node[main] (32) [below of =31] {$\mathbf{110}$};
\node[main] (33) [below of =32] {$\mathbf{111}$};
\node[main] (34)[right of =26] {$\mathbf{000}$};
\node[main] (35) [below of=34] {$\mathbf{001}$};
\node[main] (42)[right of =34] {$\mathbf{000}$};

\draw[->,very thick] (1) -- node[midway,above] {\textbf{0}} (2);
\draw[->,very thick] (1) -- node[midway,right = 0.2 cm,above] {\textbf{1}} (5);
\draw[->,very thick] (2) -- node[midway,above] {\textbf{0}} (10);
\draw[->,very thick] (2) -- node[midway,right = 0.2 cm,above] {\textbf{1}} (15);
\draw[->,very thick] (5) -- node[near start,above] {\textbf{0}} (13);
\draw[->,very thick] (5) -- node[midway,right = 0.2 cm,above] {\textbf{1}} (16);
\draw[->,very thick] (10) -- node[midway,above] {\textbf{0}} (18);
\draw[->,very thick] (10) -- node[near start,right = 0.2 cm,above] {\textbf{1}} (24);
\draw[->,very thick] (13) -- node[near start,above] {\textbf{0}} (21);
\draw[->,very thick] (13) -- node[near start,left = 0.2 cm,below] {\textbf{1}} (23);
\draw[->,very thick] (15) -- node[midway,above] {\textbf{0}} (23);
\draw[->,very thick] (15) -- node[near end,right = 0.2 cm,below] {\textbf{1}} (21);
\draw[->,very thick] (16) -- node[midway,above] {\textbf{0}} (24);
\draw[->,very thick] (16) -- node[near end,right = 0.3 cm,above] {\textbf{1}} (18);
\draw[->,very thick] (18) -- node[midway,above] {\textbf{0}} (26);
\draw[->,very thick] (21) -- node[near start,above] {\textbf{0}} (29);
\draw[->,very thick] (23) -- node[near end,left = 0.2 cm,above] {\textbf{1}} (27);
\draw[->,very thick] (24) -- node[midway,right = 0.3 cm,above] {\textbf{1}} (28);
\draw[->,very thick] (26) -- node[midway,above] {\textbf{0}} (34);

\draw[->,very thick] (27) -- node[near start,above] {\textbf{0}} (35);
\draw[->,very thick] (29) -- node[midway,right = 0.2 cm,below] {\textbf{1}} (35);

\draw[->,very thick] (28) -- node[near end,right = 0.2 cm,below] {\textbf{1}} (34);
\draw[->,very thick] (34) -- node[midway,above] {\textbf{0}} (42);
\draw[->,very thick] (35) -- node[midway,right = 0.2 cm,below] {\textbf{1}} (42);
\end{tikzpicture}
\caption{Trellis of [6,3,3] Linear Block Code Encoder. The encoder initializes 
 and terminates in the null state. There are total $8$ possible paths shown in the trellis, which satisfy the initialization and termination condition. From each state the encoder can jump into any of the two possible states depending upon the input $c_i$ at $i^\mathrm{th}$ time instant.}
\label{fig:633_trellis}
\end{figure}
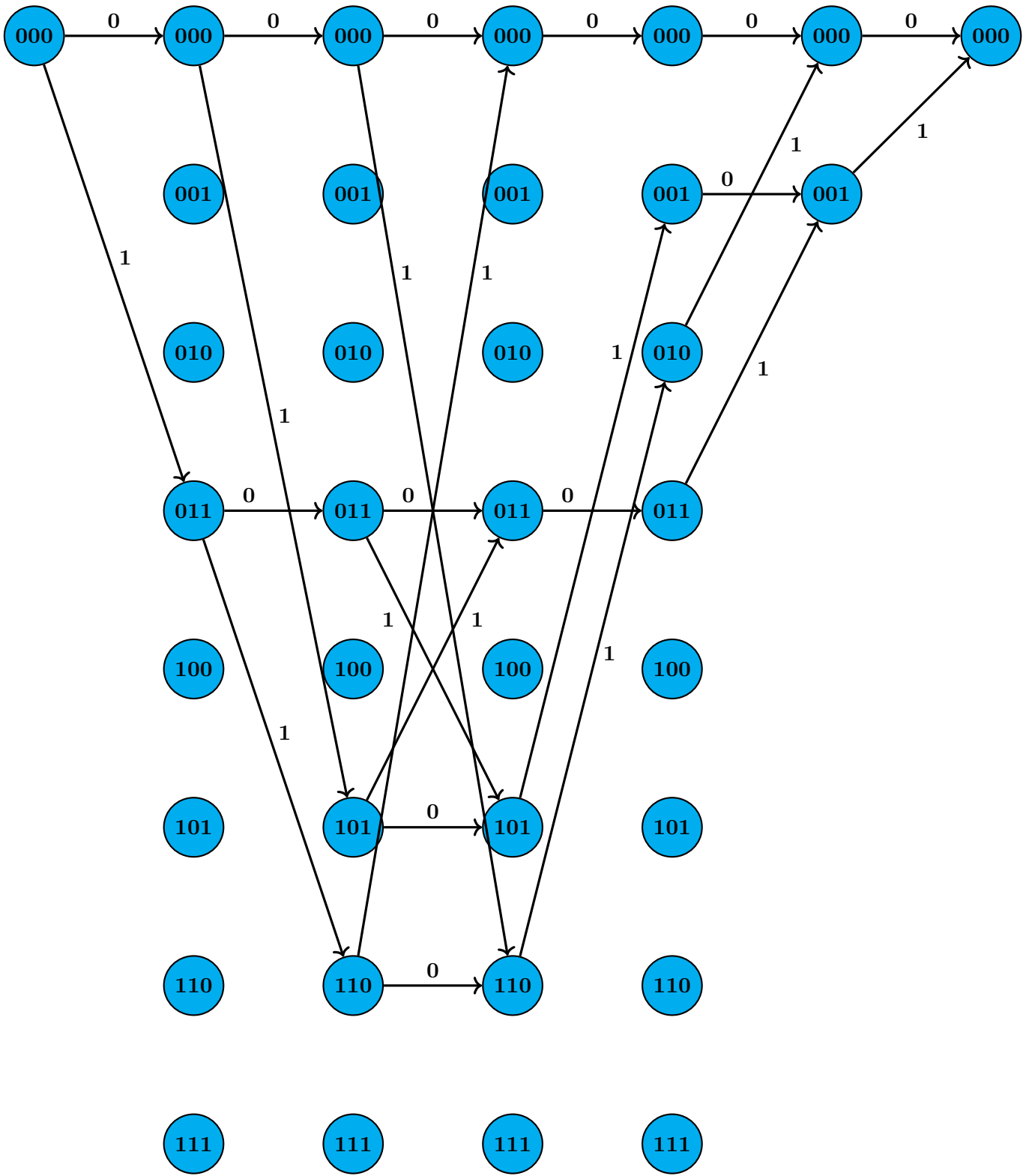
Let us assume the encoder is initialized to an all-zero state or the null state. 
Thus, Eq. \ref{eq:lbc_trellis} becomes \begin{align}\label{eq:labc_trellis2}
    S_t = \sum_{i=1}^{t}c_ih_i.
\end{align} 
Consider the parity check matrix $\mathbf{H}$ of a $[6,3,3]$ code, which we use primarily in this paper to demonstrate the algorithm: \begin{align*}
    \mathbf{H} = \begin{bmatrix}
0&1&1&1&0&0\\
1&0&1&0&1&0\\
1&1&0&0&0&1\\
\end{bmatrix}.
\end{align*}
Based on the $\mathbf{H}$, we show the trellis of the [6,3,3] code in Fig. \ref{fig:633_trellis}.
We emphasize that, due to the null state initialization, the final state of the encoder should also be a null state. 
This constraint of the trellis determines all the valid paths.
We discard the branches which do not end up in the final null state, as they do not produce a valid codeword. 
On this trellis, the Viterbi decoder performs ML decoding of linear block codes by evaluating the partial sums (Hamming distance between $i^\mathrm{th}$ trellis branch and $y_i$) of each branch in the trellis.  
The total path metric of a valid path is the total sum of all the partial sums and the path with the least total path metric is the decoded codeword.
Thus, this path thus refers to the codeword having the least Hamming distance with respect to the received erroneous vector.
Mathematically, the Viterbi decoding is expressed as finding the codeword $c^*$:
\begin{align}
    c^{*} = \argmin_{c\in \mathcal{C}}d(c,y),
\end{align}
where $r$ is the received erroneous vector and $d(c,y)$ is the Hamming distance between codeword and the received vector.
We refer to this quantity as the cost function $f(c)$, 
where, `$c$' represents the variable for various paths in the trellis.
We use the QAOA introduced by Farhi \emph{et al}. \cite{qaoa} to find the paths with least path metric in the trellis.
The cost Hamiltonian of QAOA is a map from the cost function $d(c,y)$ and further a mixer Hamiltonian is defined in Section \ref{sec:viterbi_mix}.
We show that the construction of the mixer Hamiltonian ensures preservation of the codespace $\mathcal{C}$ and depends on the parameter $d$.
We now briefly introduce the variational principle and QAOA used in the algorithm we propose.
\section{The Variational Principle and Quantum Approximate Optimization}
\label{sec:qaoa}
The variational principle has its applicability in different branches of science with a few subtle aspects.
Particularly, it is well known for the study of calculus of variations in Physics and Mathematics.
A few examples include Fermat's principle of least time, optical path near event horizon of a black hole, brachistochrone problem uses variational methods \cite{arfken2011mathematical}.
For a quantum system, sometimes the exact solution of the system Hamiltonian is not known analytically.
In those cases, one can use variational principle.
It is an approximation method for determining the ground state of the system.
The ground state is approximated by a parameterized educated guess called ansatz \cite{griffiths2017introduction}.
The expectation value of the system Hamiltonian is the cost, which is minimized by varying the parameters of the ansatz.
The upper bound for the ground state energy $E_{g}$ of the quantum system with Hamiltonian $H$ is 
\begin{align}
    \label{eq:variational}
    E_g\leq \braket{\psi(\theta)|H|\psi(\theta)},
\end{align} 
where $\psi(\theta)$ is the ansatz. 
The parameter $\theta$ is optimized to keep the expectation value in Eq. \ref{eq:variational} as close to the actual ground state energy.
This idea is extended in the hybrid quantum-classical algorithm, where the expectation value of the cost Hamiltonian is extracted from repeated calls on a quantum circuit and the optimization of parameters is done using a classical optimizer.
The QAOA is a variational quantum algorithm, where the ansatz is prepared using the PQC, which represents a circuit of time evolving unitary gates of the Hamiltonian corresponding to the optimization problem and we optimize the parameters of the PQC using a classical optimizer.
There, the parameter $\theta$ is embedded in the unitary corresponding to the problem Hamiltonians, which is coined as ``problem-inspired ansatze" by Cerezo \emph{et al.} \cite{cerezo2021variational}.
We now formally discuss the workflow of QAOA.
\subsection{Quantum Approximation Optimization Algorithm}
\label{sec:Quantum Approximation Optimization Algorithm}
Primary objective of QAOA is to encode an optimization problem into a problem of finding the ground state of a related cost Hamiltonian.
This makes it more laudable for the case of Viterbi decoding.
The ground state of a Hamiltonian refers to the state with minimum eigenvalue.
QAOA utilizes both quantum and classical methods to obtain the ground state of an observable Hamiltonian.
The optimization problem is generally expressed in terms of a Boolean or pseudo-Boolean function.
The QAOA uses the PQC to evaluate this cost function by estimating the expectation value of the corresponding Hamiltonian. 
Then, a classical optimizer is used to tweak the parameters so that we achieve minimum expectation value of the Hamiltonian.
A general schematic representation of QAOA is shown in Fig. \ref{fig:qaoa}.
Encoding the cost functions of any problem into a diagonal Hamiltonian is the first key step of QAOA. 
The cost function, in general is a map $f: \{0,1\}^n \longrightarrow \mathbb{R}$ defined by
\begin{align}
    \label{eq:cost_fn}
    f(c) = \sum_{i=1}^{\alpha} w_iC_i(c),
\end{align}
where $c$ is a $n$ bit binary string and $w_i\in \mathbb{R}$, and $C_i(c)$ is a Boolean function that takes the $n$ bit string as input and gives a binary output, i.e. $C_i: \{0,1\}^n \longrightarrow \{0,1\}$.
We note that $C_i(c)$ acts locally on some of the bits of the $n$-bit string $c$.
The representation in Eq. \ref{eq:cost_fn} suggests that the total cost function consists of a weighted sum of several local Boolean functions, where $\alpha$ is the total number of such local terms.
The locality of each Boolean function is the locality of its acting domain. 
This weighted sum representation of the cost function $f(c)$ is a pseudo-Boolean function.
For the problem of Viterbi decoding, we choose a pseudo-Boolean cost function, which evaluates the Hamming distance between two vectors.\\
The diagonal Hamiltonian corresponding to the cost function $f(c)$ should be constructed such that
\begin{align}
    \label{eq:cost_hamiltonian_map}
    H_f\ket{x} = f(c)\ket{x},
\end{align}
where $\ket{c}$ is eigenstate of $H_f$ with the eigenvalue $f(c)$.
This $H_f$ is called the cost Hamiltonian.
Eigenvalue of the cost Hamiltonian is the real value of the cost function defined in Eq. \ref{eq:cost_fn}. 
The expectation value of this Hamiltonian $\braket{H_f}$ refers to the cost that must be minimized.
In the PQC, time evolution of the cost Hamiltonian is defined as
\begin{align}
    U_f(\gamma) = e^{-i \gamma H_f},
\end{align}
where we call it the cost unitary, with $\gamma$ as the evolution parameter of $H_f$.\\
A single layer of the PQC contains another unitary along with the $U_f(\gamma)$. It is called the mixer and is defined as 
\begin{align}
    \label{eq:mixer_qaoa}
    U_m(\beta) = e^{-i \beta H_m},
\end{align}
where, $H_m$ is the mixer Hamiltonian that does not commute with $H_f$ \cite{fuchs2022constraint}. The parameter $\beta$ is the the evolution parameter of $H_m$.
The mixer Hamiltonian is used to ensure a non uniform evolution of the initial state over which $U_f(\gamma)$ and $U_m(\beta)$ is applied repeatedly.
\begin{figure}
    \centering
    \includegraphics[width=\textwidth]{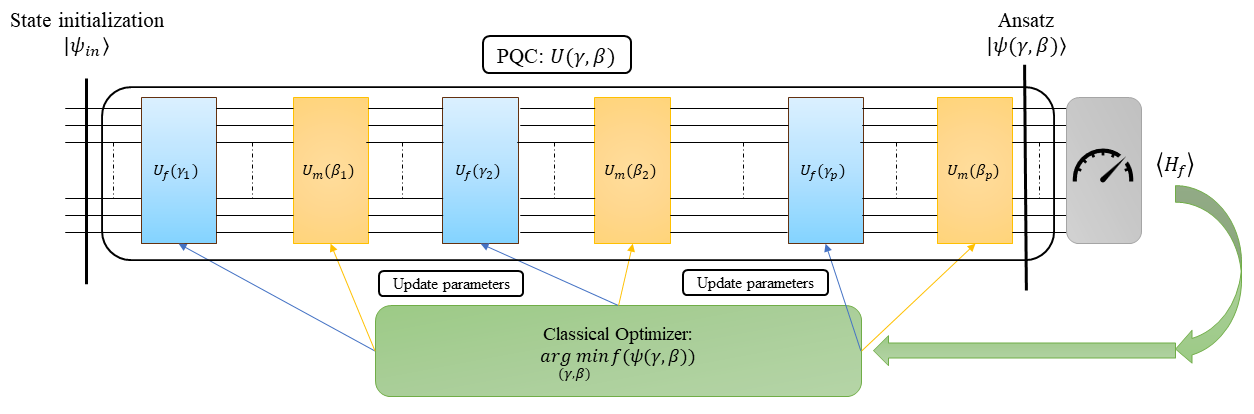}
    \caption{This is a schematic representation of QAOA. $U(\gamma, \beta)$ makes the parameterized quantum circuit (PQC), such that $U(\gamma, \beta) = U_m(\beta_p)U_f(\gamma_p)....U_m(\beta_2)U_f(\gamma_2)U_m(\beta_1)U_f(\gamma_1)$.
    The PQC thus, consists of $p$ layers of the parameterized unitary gates, and $\gamma, \beta$ are two sets of parameters; i.e. $\gamma = (\gamma_1,\gamma_2,.....\gamma_p)$ and $\beta = \beta_1,\beta_2,......,\beta_p$. 
    $\ket{\psi(\gamma, \beta)}$ is the ansatz prepared by the application of $U(\gamma, \beta)$ on the initialized state $\ket{\psi_{\mathrm{in}}}$.
    The expectation value of the cost Hamiltonian $H_f$ is minimized by a classical optimizer and gives a feedback for the updated parameters to the PQC.}
    \label{fig:qaoa}
\end{figure}
The initial state is the highest energy state of the mixer Hamiltonian $H_m$ \cite{blekos2023review}.
For example, if all the computational basis states of a $n$ qubit system are the eigenstates of the $H_f$, then the most useful initial state will be the equal superposition of all the basis states. 
This can be achieved with unit depth operation of $\mathrm{Hadamard}$ gates on $n$ qubits, giving an initial state:
\begin{align}
    \ket{\psi_{\mathrm{in}}} = \frac{1}{\sqrt{2^n}}\sum_{i=1}^{2^{n}} \ket{i}.
\end{align}
On this initialized state, we operate the parameterized unitaries $U_f$ and $U_m$ repeatedly to amplify the probability of the solution state.
The solution state is the state corresponding to the decoded codeword, which gives the minimum value for the cost function $f(c)$.
If the unitaries are being applied for $p$ times, the output state of the PQC is defined as the ansatz $\ket{\psi(\beta, \gamma)}$ and it is as follows,
\begin{align}
    \label{eq:param_unitary}
    \ket{\psi(\gamma, \beta)} = e^{-i\beta_p H_m}e^{-i\gamma_p H_f}& e^{-i\beta_{p-1} H_m}e^{-i\gamma_{p-1} H_f}\cdots \nonumber\\ \cdots 
    & \, e^{-i\beta_2 H_m}e^{-i\gamma_2 H_f}e^{-i\beta_1 H_m}e^{-i\gamma_1 H_f}\ket{\psi_{in}}.
\end{align} 
The parameters $\gamma = (\gamma_1,\gamma_2,.....\gamma_p)$ and $\beta=(\beta_1,\beta_2,......,\beta_p)$ requires optimization classically.
A properly optimized PQC outputs the decoded solution state with higher probability.\\
Eq. \ref{eq:cost_hamiltonian_map} implies $\braket{f(c)} = \braket{H_f}$, where $\braket{f}$ is the average of the classical cost function and $\braket{H_f}$ is the expectation value of the observable $H_f$ evaluated over the ansatz quantum state $\ket{\psi(\gamma, \beta)}$.
Preparing the ansatz state of Eq. \ref{eq:param_unitary} in a quantum computer multiple times and measuring the output gives an estimate of the expectation value 
\begin{align}
    \label{eq:expn}
    \braket{\psi(\gamma, \beta)|H_f|\psi(\gamma, \beta)} = \sum_{i=1}^{2^n}p_if(i),
\end{align}
where $p_i$ denotes the probability of the occurrence of $i^\mathrm{th}$ basis state.
Minimizing this expectation value using a classical optimizer, gives us the optimized parameters $\gamma^*, \beta^*$ for the PQC.
The optimized parameters when used again according to Eq. \ref{eq:param_unitary} on a quantum computer, leads to an increment in the probability of the solution state compared to the non-solution states.
The success of the above optimization not only depends on finding the good set of parameters namely $\gamma^*, \beta^*$ but also on the choice of the cost function as well, which we discuss next.
\subsection{Works related to the hurdles of QAOA}
\label{subsec:Works related to the hurdles of QAOA}
Often QAOA fails to produce optimized solution to a problem \cite{trapps}.
Problem-efficient ansatzes are used in PQC due to their simplicity and resource efficient implementation on hardware compared to the circuits compiled from pure quantum algorithms.
But random initialization of parameters in the PQC is a major factor behind the emergence of barren plateaus \cite{mcclean}. 
A barren plateau refers to a flat optimization landscape, which is quantified by exponentially vanishing variance of the gradients of the cost function. 
In their work, McClean \emph{et al}. showed that these barren plateaus occur for deep versions of the randomly initialized PQC, essentially the depth $D$ is at least $D = \mathcal{O}(\mathrm{poly}(n))$ for the occurrence of these plateaus \cite{mcclean}.
Deep versions of PQC refers to circuits having higher $p$.
A flatter optimization landscape for deep circuits suggests that the success in finding a good solution strictly depends on the initial guess of the parameters.
As far as we know, there are no optimal techniques for determining good parameters beforehand.
In their work, Cerezo \emph{et al.} showed that the occurrence of barren plateaus is also cost function dependent \cite{cerezo}. 
The authors showed that if the cost function Hamiltonian is global, then barren plateaus will occur even for shallow-depth circuits and there is no escape from it.
However, the gradients vanish at worst polynomially for lower depths up to a depth of $D = \mathcal{O}(\log(n))$, if the defined cost function Hamiltonian is local. 
They also proved that the gradients vanish exponentially if the depth is $D = \mathcal{O}(\mathrm{poly}(n))$ and beyond. 
This renders the PQC trainable for shallow circuits. 
We use this fact to construct a trainable parameterized quantum circuit for our hybrid Viterbi decoder. 
\\\\
The inevitability of barren plateaus in higher-depth regions is worrying because this is the region where we get the solutions of the algorithm with high probability.
We formulate a uniform parameter optimization technique, which gives us efficient results for finding good solutions even for lower-depth PQC.
This is an approach to avoid barren plateaus and also makes the PQC of the hybrid decoder efficient in the context of NISQ devices.
Also, we formulate the cost function such that its Hamiltonian acts locally. 
This makes our PQC more trainable and gives us an efficient quantum algorithm for Viterbi decoding. 
We now formally develop the hybrid Viterbi decoder.
\section{Quantum Classical Hybrid Viterbi Decoder}
\label{sec:Quantum Classical Hybrid Viterbi Decoder}
We map the Viterbi decoding problem into its variational quantum counterpart. 
To design the quantum circuit which  essentially evaluates the cost function, we need to construct the cost and the mixer Hamiltonian such that they do not commute. 
In Eq \ref{eq:param_unitary}, their non-commuting nature ensures a change in amplitudes of the initial state on time evolution.
\subsection{Construction of the Cost Hamiltonian}
\label{subsec:Construction of the Cost Hamiltonian}
Consider the codeword $c=(c_1,c_2,\cdots,c_n)$ and received vector $y=(y_1,y_2,\cdots, y_n)$ such that $c,y \in \{0,1\}^n$ as considered in Section \ref{sec:Viterbi decoding of classical codes}.  Based on the discussion in Section \ref{sec:Quantum Approximation Optimization Algorithm}, we propose the cost function for the Viterbi decoder as \begin{align}
    f(c) = \sum_{i=1}^{n}c_i \oplus y_i,
    \label{eq:cost_func}
\end{align}
where $\oplus$ is bit-wise the modulo $2$ addition.
We note that $f$ is a real-valued function, written as a weighted sum of Boolean functions of the form given in Eq. \ref{eq:cost_fn}. 
We also note that the cost function consists of finite number of local operations.
The cost function defined in Eq. \ref{eq:cost_func} takes the $n$ bit codeword as input and returns the Hamming weight between the codeword and the received vector, as the output.
The unique diagonal Hamiltonian representation of the pseudo-Boolean function is represented by a weighted sum of Pauli-$\mathrm{Z}$ operators.
The terms in the summation correspond to the Fourier expansion of the function \cite{boolean_func}.
Each bit of the codeword $c_i$ is mapped into the corresponding qubit denoted as $x_i, i \in \{1,2,.... n\}$ in the quantum register $x$.
Each bit of the received vector $y_i$ is mapped into the corresponding qubit denoted by $r_i, i \in \{1,2,.... n\}$ in the quantum register $r$.
The Hamiltonian corresponding to Eq. \ref{eq:cost_func} is \begin{align}
    H_f = \sum_{i=1}^{n} \left(\frac{1}{2}\mathrm{I} - \frac{1}{2}{\mathrm{Z}_{x_i}}{\mathrm{Z}_{r_i}}\right),
    \label{eq:cost_hamiltonian}
\end{align}
where
$\mathrm{Z}_{x_i}$ is the Pauli-$\mathrm{Z}$ operator acting on the $i^\mathrm{th}$ qubit of the register $x$. 
Similarly, $\mathrm{Z}_{r_i}$ is the Pauli- $\mathrm{Z}$ operator acting on the $i^\mathrm{th}$ qubit of the register $r$.
Also, we would like to point out that the register $r$ is initialized according to the received erroneous vector.
This refers to $r_i$ being fixed spins and the spin is in accordance with the value of $y_i$.
In Appendix \ref{ap:appendixproof} we showed how we came up with Eq. \ref{eq:cost_hamiltonian}.
The goal is to find the ground state of the above cost Hamiltonian, while restricting ourselves only to the codespace $\mathcal{C}$.
For the same, QAOA requires a mixer Hamiltonian. 
We propose a correlation between the minimum Hamming distance of the code and the mixer Hamiltonian next, that will ensure the preservation of the above mentioned constraint.
\subsection{Construction of the Mixer Hamiltonian}
\label{sec:viterbi_mix}
Along with the cost unitary, a mixer unitary is also applied on the initial state. 
The states in the initial uniform superposition are the eigenstates of the cost Hamiltonian. 
So, only the cost unitary cannot disturb the amplitude profile of the initial state. 
It only adds a phase to each of the states in the superposition without altering the probability profile of the individual states.
The mixer Hamiltonian in that sense mixes the amplitudes.
The mixer Hamiltonian corresponding to the mixer unitary must conserve the search space, which is essential in the proposed Viterbi decoding algorithm. 
In many applications of QAOA including the Viterbi decoding, the search space does not contain all the basis members of a $n$ qubit system. 
The restricted search space in such situations requires choosing the mixer Hamiltonian that does not map the states outside of the search space $\mathcal{C}$. 
Therefore, the mixer unitary must conserve the same search space.
The constraint on the restricted search space determines the state initialization, as discussed in Section \ref{sec:Quantum Approximation Optimization Algorithm}.
Therefore, the initial state is
\begin{align}
    \ket{\psi_{\mathrm{in}}} = \frac{1}{\sqrt{2^k}}\sum_{i \in \mathcal{C}} \ket{i}.
    \label{eq:cons_ini}
\end{align}
Further, the mixer Hamiltonian must not have any eigenstates in the codespace either.
This ensures the non-commuting nature of the mixer Hamiltonian and the cost Hamiltonian.
Now we focus on constructing the mixer Hamiltonian, compatible with the properties it must satisfy, as discussed. 
We assume there must be a null element in the search space, which is the all zero codeword.
Further we consider there are $v$ elements in the search space (essentially the codespace) having Hamming distance $d$ with respective to the null element, where $d$ is the minimum Hamming distance of the codespace.
This number $v$ is different for different codes.
We reiterate that finding the codewords with Hamming weight $d$ is a NP-complete problem as shown by Berlekamp \emph{et al.} \cite{berlekamp}.
A Grover inspired algorithm can be used as a subroutine to introduce a quadratic improvement in finding those codewords \cite{syndrome_us}.
Considering we have the $d$ distance codewords at ease, we state a theorem to generate the valid mixer for any restricted search space.
\begin{theorem}
\label{theorem_mixer}
If $\mathcal{C}$ is a restricted search space of a linear block code in a $2^n$ dimensional system and $d$ is the minimum Hamming distance of $\mathcal{C}$, then the mixer Hamiltonian $H_m$ is the sum of d-local Pauli-$\mathrm{X}$ operators, which maps the null element to any of the $v$ elements in $\mathcal{C}$, with Hamming weight $d$: 
\begin{align}
    \label{eq:mixer_hamiltonian}
    H_m = \sum_{c: w_H(c) = d}\,\,\prod_{j: c_j = 1}\mathrm{X_j}.
\end{align}
\end{theorem}

\begin{proof}
To prove Theorem \ref{theorem_mixer}, we first define the framework of the mixer Hamiltonian. 
Let $\ket{\psi}$ be the initial state of the variational algorithm. 
We choose it to be the equal superposition of all the states from the codespace. 
This makes the initial guess feasible. The initial guess is 
\begin{align}
    \ket{\psi_{\mathrm{in}}} & = U_g \ket{0}^{\otimes n}, \nonumber\\
    & = \frac{1}{\sqrt{2^k}} \sum_{i \in \mathcal{C}} \ket{i},
\end{align}
where, $U_g$ is an encoding operator ensuring Eq. \ref{eq:cons_ini} \cite{syndrome_us}. 
The mixer is defined as $e^{-i\beta H_m}$ in Eq. \ref{eq:mixer_qaoa}. 
The mixer must satisfy the constraint of mapping the states into the codespace only and as $H_m$ is Hermitian, the mapping of both the mixer and the mixer Hamiltonian is same.
\begin{align}
    \label{eq:span of mixer}
    & U_m\ket{\psi_{\mathrm{in}}}  \in \mathcal{C}, \nonumber\\
    \implies & e^{-i\beta H_m}\ket{\psi_{in}}  \in \mathcal{C}, \nonumber\\
      \implies & H_m\ket{\psi_{in}}  \in \mathcal{C}.
\end{align}
A detailed proof of this equivalence is given in Appendix \ref{ap:appendix1}.
To follow this constraint we propose the Hamiltonian be
\begin{align}
    \label{eq:defined mixer}
    H_m = \sum_{c,c^{\prime} \in \mathcal{C}}g_{c,c^{\prime}}\ket{c}\bra{c^{\prime}}, 
\end{align}
\begin{align}
\label{eq:g}
\text{where,} \hspace{0.4cm} g_{c,c^{\prime}}  =
\begin{cases}
     1 &\forall \, (c,c^{\prime}): d(c,c^{\prime}) = d. \\
     0 &\text{otherwise}.
\end{cases}
\end{align}
Eq. \ref{eq:defined mixer} and Eq. \ref{eq:g} suggests that $H_m$ must contain the tensor product of $d$-local Pauli operators, such that individually they map any state to a state with Hamming distance $d$. $g$ is the matrix which determines this valid transitions through Eq. \ref{eq:g}, and has a dimension $\mathrm{dim}(\mathcal{C})\times\mathrm{dim}(\mathcal{C})$. 
Thus for any given $(n,k,d)$ code $g$ can be constructed. 
Using the matrix $g$ and expanding the matrices for each term of Eq. \ref{eq:defined mixer}, it is verified that 
\begin{align}
    H_m = \sum_{c,c^{\prime} \in \mathcal{C}}g_{c,c^{\prime}}\ket{c}\bra{c^{\prime}} = \sum_{c: w_H(c) = d} \, \, \prod_{j: c_j = 1}\mathrm{X_j} .
\end{align}
Now for any linear block code using $\ket{c} = \ket{0}^{\otimes n}$, Eq. \ref{eq:defined mixer} implies:
\begin{align}
    H_m\ket{0}^{\otimes n} & = \sum_{c^{\prime} \in \mathcal{C}}g_{0,c^{\prime}}\ket{c^{\prime}}\bra{0}^{\otimes n}\ket{0}^{\otimes n}, \nonumber\\
    & = \sum_{c^{\prime}: d(c,c^{\prime}) = d} \, \ket{c^{\prime}}.
\end{align}
\end{proof}
Thus, the expression of $H_m$ in Eq. \ref{eq:mixer_hamiltonian} is verified. 
Each term of Eq. \ref{eq:mixer_hamiltonian} or Eq. \ref{eq:defined mixer} is mapping the elements in $\mathcal{C}$ into another element in the same $\mathcal{C}$, with Hamming distance $d$.
From the Eq. \ref{eq:defined mixer} it is fairly clear that the depth of the PQC increases according to the minimum Hamming weight and structure of the code. 
If in the codespace there are more elements with minimum Hamming weight respective to the null element, the circuit depth increases. 
Thus, PQC for different linear block codes may have higher depths. 
We now focus on incorporating these two Hamiltonians into their corresponding unitary gates and construct the parameterized quantum circuits for decoding linear block codes using the Viterbi decoding method. 

\subsection{Parameterized Quantum Circuit design for [6,3,3] Linear Block Code}
\label{subsec:Parameterized Quantum Circuit}
\begin{figure}[t!]
    \centering
\begin{quantikz}[font = \Huge]
\lstick{$\ket{x_1}$} & \ctrl{1} & \qw & \ctrl{1} & \qw\\
\lstick{$\ket{r_1}$} & \targ{} & \gate{R_\mathrm{z}(-\gamma)} & \targ{} & \qw
\end{quantikz}
    \caption{Circuit for $e^{i\frac{\gamma}{2}{\mathrm{Z}_{x_1}}{\mathrm{Z}_{r_1}}}$. 
    This is one of the terms from Eq. \ref{eq:cost_unitary}. 
    $\ket{x}$ refers to the codeword register and $\ket{r}$ refers to the ancilla register. $R_z(-\gamma)$ is the single qubit rotation gate about the $z$ axis, $R_z(-\gamma) = e^{i \frac{\gamma}{2} \mathrm{Z}}$.}
    \label{fig:exp_z1z2}
    \centering

\begin{quantikz}[font = \Huge]
    \lstick{$\ket{x_1}$} & \gate{\mathrm{H}} & \ctrl{1} & \qw & \qw & \qw & \ctrl{1} & \gate{\mathrm{H}} & \qw\\
    \lstick{$\ket{x_2}$} & \gate{\mathrm{H}} & \targ{} & \ctrl{1} & \qw & \ctrl{1} & \targ{} & \gate{\mathrm{H}} & \qw\\
    \lstick{$\ket{x_3}$} & \gate{\mathrm{H}} & \qw & \targ{} & \gate{R_\mathrm{z}(2\beta)} & \targ{} & \qw & \gate{\mathrm{\mathrm{H}}} & \qw
\end{quantikz}
    \caption{This is the circuit for $e^{-i \beta \mathrm{X}_1\mathrm{X}_2\mathrm{X}_3}$. We exploit the fact that $H \otimes H \otimes H e^{i \beta \mathrm{Z}_1 \mathrm{Z}_2 \mathrm{Z}_3} H \otimes H \otimes H = e^{i \beta \mathrm{X}_1 \mathrm{X}_2 \mathrm{X}_3}$.}
    \label{fig:expx0x1x2}
\end{figure}
We now design the PQC for the $[6,3,3]$ code. 
The initial state is an equal superposition of all the codewords for the $[6,3,3]$ code. 
Thus, the restricted search space is of dimension $2^3$. 
We can generate this superposition using a unitary operator constructed from the generator matrix of the $[6,3,3]$ code \cite{qml}. The generator matrix of this $[6,3,3]$ code is
\begin{align}
    \label{eq:generator633}
    \mathbf{G} = \begin{bmatrix}
1&0&0&0&1&1\\
0&1&0&1&0&1\\
0&0&1&1&1&0\\
\end{bmatrix}.
\end{align}
The unitary  operator $U_g$ \cite{qml}, which generates an equal superposition of all codewords of the $[6,3,3]$ code belonging to codespace $\mathcal{C}$ is
\begin{align}
    U_g = \prod_{{4}\leqslant i\leqslant {6}}\Big(\prod_{1\leqslant j\leqslant {3},G_{ji}=1} \mathrm{CX}_{j,i}\Big) \prod_{1\leqslant j\leqslant {3}} \mathrm{H}_j,
\end{align}
where $\mathrm{CX_{j,i}}$ is the controlled Pauli-$\mathrm{X}$ gate.
Using this $U_g$ on a set of all qubits initialized to $\ket{0}$ generates the desired uniform superposition of the elements in codespace $\mathcal{C}$. The initial state can be written as
\begin{align}
    \label{eq:initial_state}
    \ket{\psi_{\mathrm{in}}} & = U_g \ket{0}^{\otimes 6},\nonumber \\
               & = \frac{1}{\sqrt{2^3}} \sum_{i \in \mathcal{C}} \ket{i}.
\end{align} 
We require additional six ancilla qubits which store the information of the received erroneous vector. 
All the ancilla qubits are initialized to the $\ket{0}$ state. 
Whenever a bit of the received vector is $1$, we apply a Pauli-$\mathrm{X}$ gate on the corresponding qubit.
The reason we incorporate these ancillas lies in the structure of the cost function Hamiltonian $H_f$ in Eq. \ref{eq:cost_hamiltonian}.
The unitary corresponding to the $H_f$ is
\begin{align}
    U_f(\gamma) & = e^{-i\gamma H_f}, \nonumber \\
        & = e^{-i\gamma \sum_{i=1}^{n} \left(\frac{1}{2}\mathrm{I} - \frac{1}{2}{\mathrm{Z}_{x_i}}{\mathrm{Z}_{r_i}}\right)}, \nonumber \\
        & = \prod_{i = 1}^{n}e^{i\frac{\gamma}{2}{\mathrm{Z}_{x_i}}{\mathrm{Z}_{r_i}}},
        \label{eq:cost_unitary}
\end{align}
where $n=6$ is the number of bits in the codeword and $\gamma$ is a scalar evolution parameter of $H_f$.
The identity operator contributes a global phase to the state. 
As the global phase has no physical effect on the quantum state, the effect of the identity operator is redundant in the PQC.
Eq. \ref{eq:cost_unitary} essentially represents a connection 
between the respective qubits of register $x$ and
$r$. 
For example, we demonstrate the quantum circuit of one such term $e^{i\frac{\gamma}{2}{\mathrm{Z}_{x_1}}{\mathrm{Z}_{r_1}}}$ from Eq. \ref{eq:cost_unitary} in Fig. \ref{fig:exp_z1z2}. 
This connects $x_1$ to $r_1$.\\
The other unitary we need to apply with the cost unitary is the mixer $H_m$. 
For which we first describe the mixer Hamiltonian for the $[6,3,3]$ linear block code.
The $[6,3,3]$ code generated by the generator matrix of Eq. \ref{eq:generator633} is\begin{align*}
\mathcal{C} = \{000000, 001110, 010101, 100011, 011011, 110110, 101101, 111000\}.
\end{align*}
Using Eq. \ref{eq:mixer_hamiltonian} from Theorem \ref{theorem_mixer} we find that the mixer Hamiltonian for the $[6,3,3]$ code is 
\begin{align}
    H_m = \mathrm{X}_1\mathrm{X}_2\mathrm{X}_3+\mathrm{X}_1\mathrm{X}_5\mathrm{X}_6+\mathrm{X}_3\mathrm{X}_4\mathrm{X}_5+\mathrm{X}_2\mathrm{X}_4\mathrm{X}_6.
\end{align}
Thus, the mixer unitary (mixer) defined as $U_m(\beta) = e^{-i\beta H_m}$ is of the form
\begin{align}
    U_m(\beta) & = e^{-i\beta (\mathrm{X}_1\mathrm{X}_2\mathrm{X}_3 + \mathrm{X}_1\mathrm{X}_5\mathrm{X}_6 + \mathrm{X}_3\mathrm{X}_4\mathrm{X}_5 + \mathrm{X}_2\mathrm{X}_4\mathrm{X}_6)}, \nonumber \\
        & = e^{-i \beta \mathrm{X}_1\mathrm{X}_2\mathrm{X}_3}  e^{-i \beta \mathrm{X}_1\mathrm{X}_5\mathrm{X}_6}  e^{- i \beta \mathrm{X}_3\mathrm{X}_4\mathrm{X}_5}  e^{-i \beta \mathrm{X}_2\mathrm{X}_4\mathrm{X}_6},
        \label{eq:mixmix}
\end{align}
where $\beta$ is a scalar evolution parameter of $H_m$.
\begin{figure}[t!]
    \centering
    \includegraphics[height=11cm,width=0.75\textwidth]{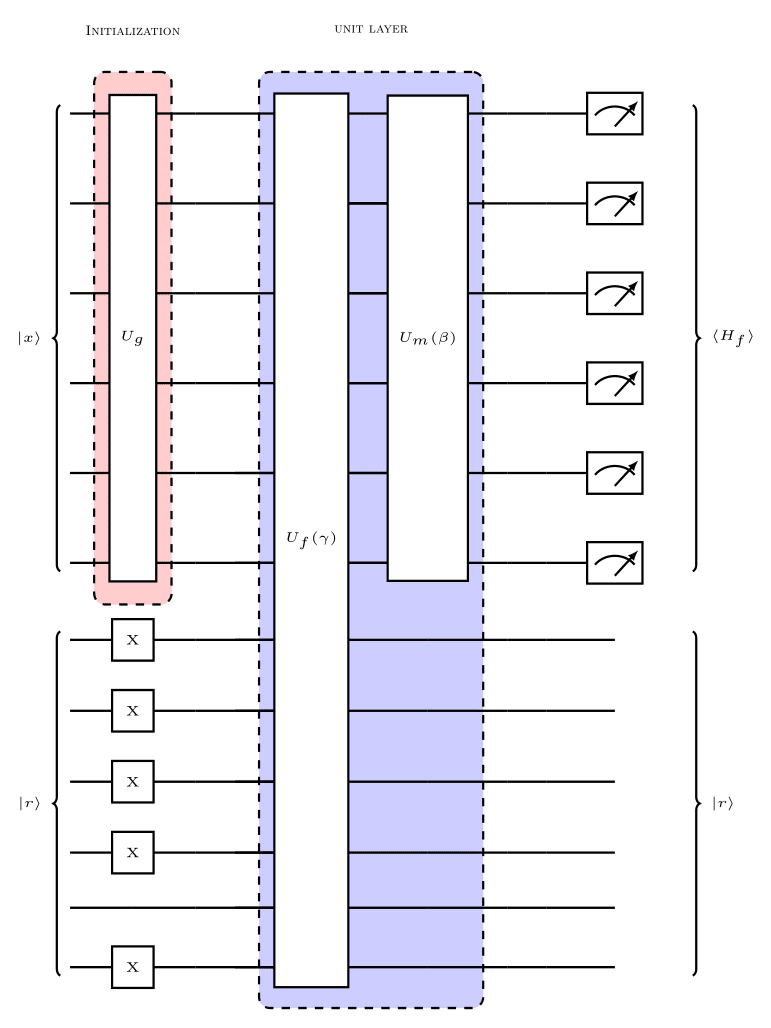}
    \caption{This is the circuit for an unit layer ($p = 1$) of QAOA applied on the initial state of the PQC. 
    $\ket{x}$ is the register representing $c$. 
    $\ket{r}$ is the register representing $y$. 
    We show here the circuit with $\ket{r} = \ket{111101}$. 
    $U_g$ creates the valid initial state $\ket{\psi_{in}}$ \cite{syndrome_us} of Eq. \ref{eq:initial_state}. An unit layer of QAOA consists of a single application of the cost unitary $U_f(\gamma)$ Eq. \ref{eq:cost_unitary} and the mixer $U_m(\beta)$ Eq. \ref{eq:mixmix}. 
    The mixer unitary is applied on the $x$ register only, following the Theorem \ref{theorem_mixer}. The detailed circuit for the mixer and cost unitary is shown in Appendix \ref{subsec:Circuit of the mixer} and Appendix \ref{subsec:Circuit of the cost unitary} respectively. Expectation value of the cost Hamiltonian is obtained by Measuring the $x$ register through repeated calling of the PQC.}
    \label{fig:unitlayer633}
\end{figure}
We demonstrate the quantum circuit for the term $e^{-i \beta \mathrm{X}_1\mathrm{X}_2\mathrm{X}_3}$ from Eq. \ref{eq:mixmix} in Fig. \ref{fig:expx0x1x2}. Following this pattern of $\mathrm{CNOT}$, $\mathrm{Hadamard}$ and $\mathrm{R_Z}$ gate we construct circuits for the other exponents of the product of Pauli-$\mathrm{X}$ gates.
Finally, $U_f$ of Eq. \ref{eq:cost_unitary} and $U_m$ of Eq. \ref{eq:mixmix} form the unit layer of the PQC of QAOA for the Viterbi decoder of the $[6,3,3]$ code. 
We are referring to a single operation of the cost unitary $U_f$, Eq. \ref{eq:cost_unitary} and a single operation of the mixer, Eq. \ref{eq:mixmix} as a single unitary layer of QAOA. 
This unit layer circuit is shown in Fig. \ref{fig:unitlayer633}. 
This one layer of unitary is repeated on the initial quantum state $\ket{\psi_{in}}$ of Eq. \ref{eq:initial_state}.
If repeated $p$ times, the final state becomes $\ket{\psi(\gamma_1,\gamma_2,....,\gamma_p, \beta_1,\beta_2,....,\beta_p,)}$, as given in Eq. \ref{eq:param_unitary}. 
The PQC requires a classical optimizer to optimize these parameters for finding the minimum of the cost function. 
But as discussed in Section \ref{subsec:Works related to the hurdles of QAOA}, finding good parameters for the PQC is difficult. 
To overcome this difficulty, we propose an uniform parameter optimization method, which has better efficiency for finding good parameters in shallow depth.

\section{Proposed Uniform Parameter Optimization (UPO) Strategy}
\label{sec:UPO}
Due to the hurdles we discussed in Section \ref{subsec:Works related to the hurdles of QAOA}, it is important to develop strategies for finding good parameters of the PQC. 
This ensures minimization of the cost function given in Eq. \ref{eq:expn}.
There are two such strategies discussed by Lee \emph{et al}. \cite{fixed_parameter}.
One of them is the random initialization method, which assumes multiple random initialization of the parameters to evaluate an average of the minimum approximation ratio, which is defined as 
\begin{align}
    \label{eq;ratio}
    \alpha = \frac{f(\psi(\gamma, \beta))}{f_{\mathrm{min}}},
\end{align}
where $\psi(\gamma, \beta)$ is the string representation of the state $\ket{\psi(\gamma, \beta)}$ prepared by the PQC.
This random initialization is repeated and the set of parameters producing the least value of average $\alpha$ is chosen as the best fit for the set of optimized parameters.
Another strategy the authors discussed is parameters fixing strategy, which we denote as FPO.
In FPO at each $p$, the parameters are optimized for $q$ random initialization. 
Parameters giving minimum $\alpha$ are chosen as the optimized parameters for that particular $p$ and then $p$ is increased.
We show that this approach of FPO offers very little improvement in finding solutions.\\
We suggest training the PQC with each parameter of different layers being equal for the mixer and the cost unitaries respectively. 
It implies 
\begin{align}
    \label{eq:upo_parameter}
    \beta_1 = \beta_2 = \cdots = \beta_p = \beta_0 \text{ and } \gamma_1 = \gamma_2 = \cdots  = \gamma_p = \gamma_0. 
\end{align}
We call this uniform parameter optimization (UPO). 
\begin{algorithm}[ht!]
\caption{Proposed Uniform Parameter Optimization for the Hybrid Viterbi Decoder}\label{alg:algorithm1}
\begin{algorithmic}[1]

\State Initialize the state is uniform superposition of all the elements of the codespace $\mathcal{C}$. 
\begin{align}
\label{eq:in_state}
    \ket{\psi_{\mathrm{in}}} & = U_g \ket{0}^{\otimes n} \nonumber \\
    & = \frac{1}{\sqrt{2^k}} \sum_{i \in \mathcal{C}} \ket{i}
\end{align}

\State Add ancilla qubits and apply Pauli- $\mathrm{X}$ gates in appropriate positions to represent the received erroneous vector $y$. \begin{align*}
    \ket{\psi_{\mathrm{in}}} & \to \ket{\psi_{\mathrm{in}}} \otimes \ket{y}\\
    \ket{\psi_{\mathrm{in}}} & \to \ket{\psi_{\mathrm{in}}}\otimes\prod_{j: wt(r_j) = 1} \mathrm{X}_{r_j} \ket{0}^{\otimes n}
\end{align*}
\State Declare $f, t_{\beta}$, $t_{\gamma}$ as array of variables.
\While {$k \leq q$}
\State $\beta_o$ $\to$ random, $\gamma_o$ $\to$ random.
\State $\ket{\psi_{\mathrm{in}}} \otimes \ket{r} \to \ket{\psi(\gamma, \beta)} \otimes \ket{r} = \left(U^p_m(\beta_o)U^p_f(\gamma_o)\ket{\psi_{\mathrm{in}}}\right) \otimes \ket{r}$.
\State Measure the $n$ codespace qubits of $x$.
\State Run the circuit for $\mathrm{2000}$ shots.
\State Use Eq. \ref{eq:expn} and Eq. \ref{eq:cost_func} to find the expectation value $\braket{f} = \braket{H_f} = \sum_{i \in \mathcal{C}} \frac{m_i}{\mathrm{shots}}f(i)$.
\State Optimize $\braket{f}$ and store the optimized parameter in $t_{\beta}$ and $t_{\gamma}$ respectively and the optimized expectation value in $f$.
\EndWhile
\For {$k$ in $f[k]$}
\If {$f[k]: \mathrm{min}\left(\frac{f[k]}{f_{\mathrm{min}}}\right)$}
\State $p_1 = t_{\beta}[k]$ and $p_2 = t_{\gamma}[k]$.
\EndIf
\EndFor, \Return \hspace{0.25cm}{$p_1$, $p_2$}.
\end{algorithmic}
\end{algorithm}
We train the parameterized quantum circuit for a fixed $p$ with all the $p$ mixer parameters having the value $\beta_o$ and all the $p$ cost unitary parameters having the value $\gamma_o$. 
There are $q$ random initialization and each $q$ set of $\{\beta_o,\gamma_o\}$ is applied to evaluate the cost $\braket{H_f}$ over the $\ket{\psi(\gamma, \beta)}$, where the set of parameters $\beta$ and $\gamma$ satisfies Eq. \ref{eq:upo_parameter}.
Then we choose the parameters, which produces the least of the minimum approximation ratio $\alpha$ as the set of optimized parameters.
We observe that the above method, namely UPO, can potentially obtain better parameters for the PQC than the parameter fixing strategy or FPO. 
We present the detailed uniform parameter optimization method for a generalized $[n,k,d]$ linear block code in Algorithm \ref{alg:algorithm1}.  
\subsection{Performance Analysis of the UPO}
\label{sec:Performance Analysis of the UPO}
\begin{figure}[ht!]
     \centering
     \begin{subfigure}[]{0.48\textwidth}
         \centering
         \includegraphics[width = \textwidth]{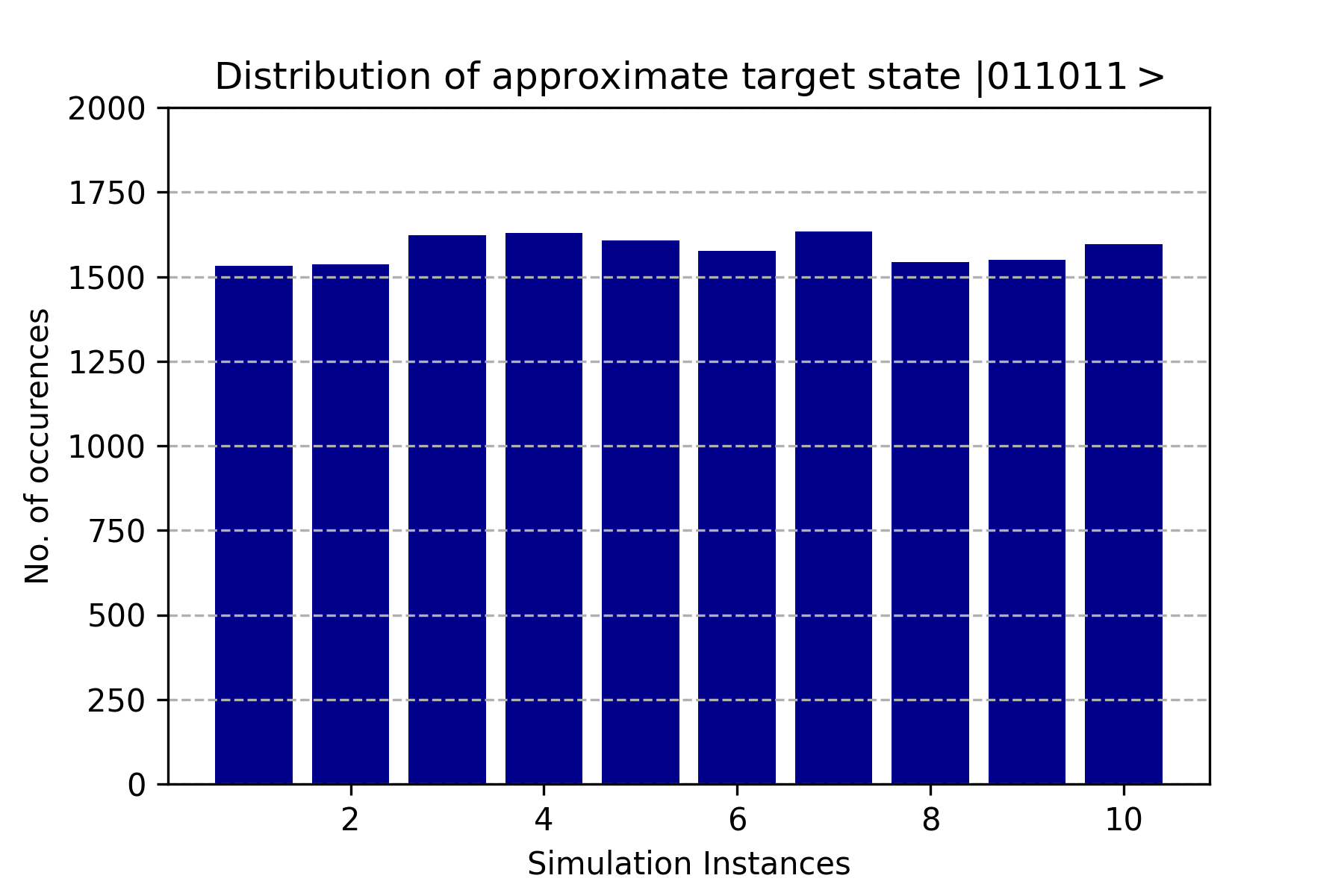}
         \caption{UPO}
         \label{fig: upo_633_10}
     \end{subfigure}
      \hfill
     \begin{subfigure}[]{0.48\textwidth}
         \centering
         \includegraphics[width = \textwidth]{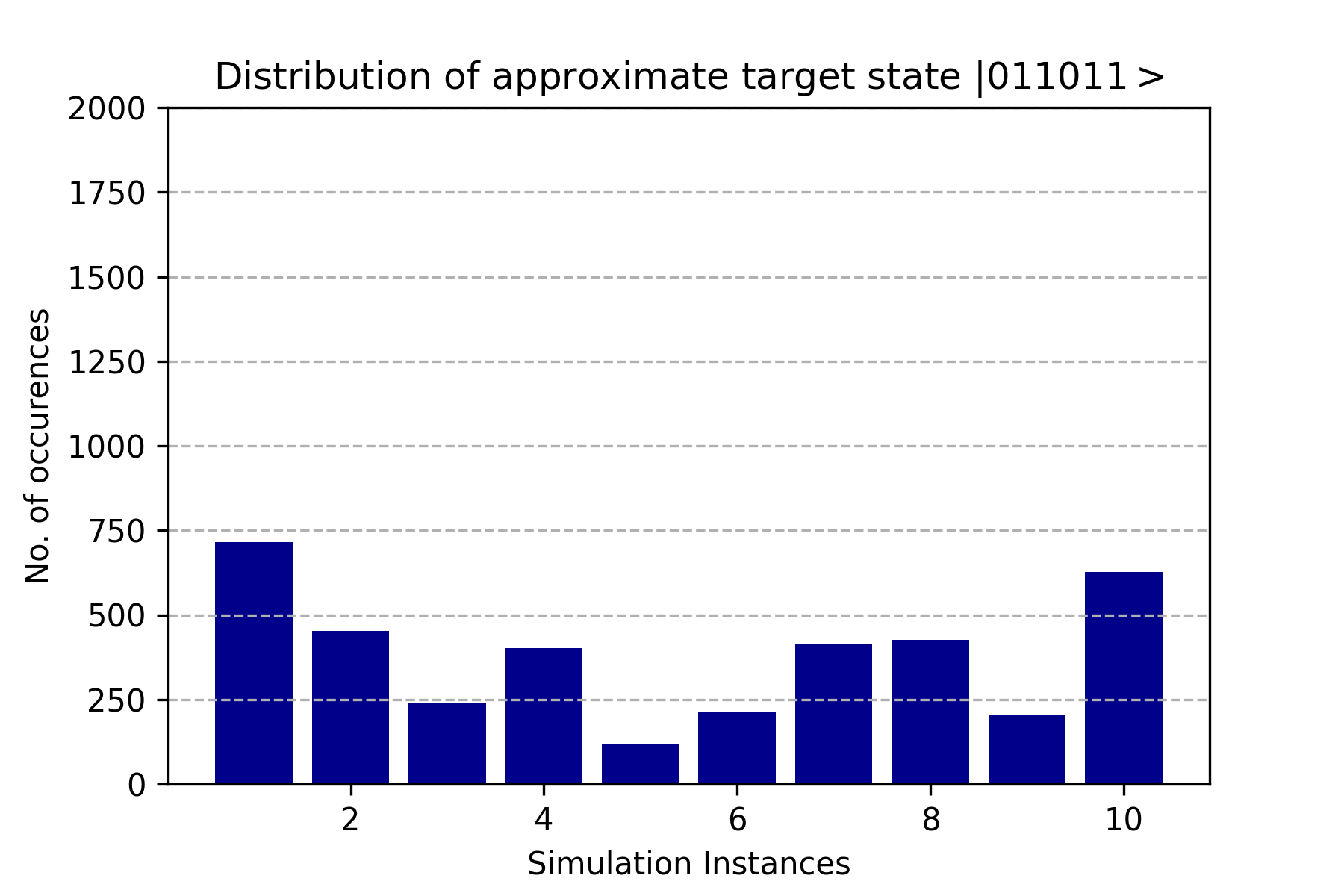}
         \caption{FPO}
         \label{fig:fpo_633_10}
     \end{subfigure}
     
        \caption{Here, in each case, the number of unitary layers is fixed to $p = 3$ for the decoder of a $[6,3,3]$ code. We choose to select the parameters with minimum approximation ratio as the good parameters out of $q = 5$ randomly initialized optimizations. We repeat each training method $10$ times and run the PQC corresponding to each set of good parameters. We observe that upon tracing the nature of the probability profile of the actual solution state, the Uniform Parameter Optimization has better potential towards obtaining good parameters, as it amplifies the actual solution state and the pattern is outstandingly regular.}
        \label{fig:upo vs fpo 1}
\end{figure}
\begin{figure}[h!]
    \centering
     \includegraphics[width = 0.48\textwidth]{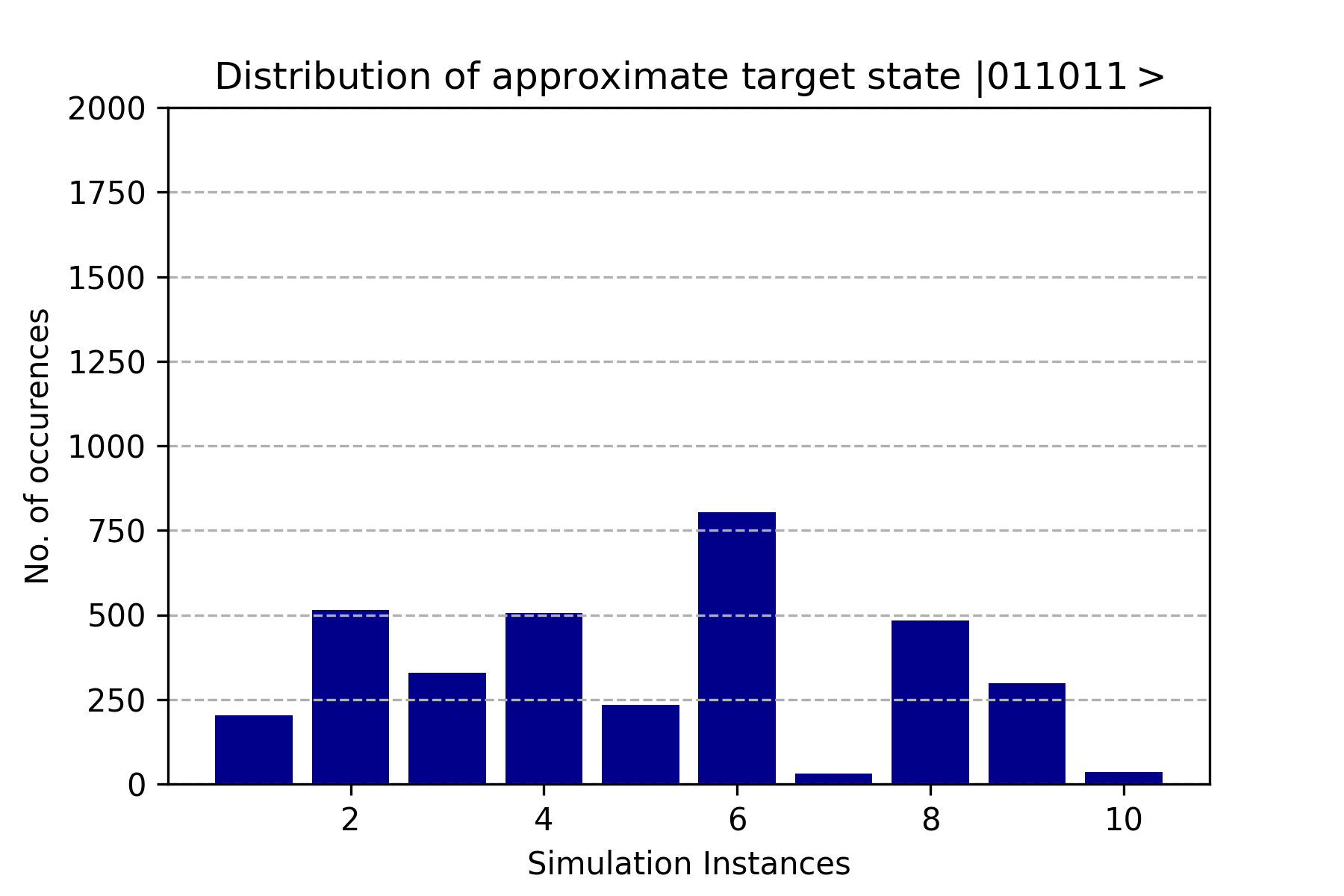}
     \caption{FPO for the $[6,3,3]$ code with $p = 3$. We chose to select the parameters which yield minimum approximation ratio as the good parameters out of $q = 20$ optimizations for each $p$. The occurrence of the state is quite random and thus is inconclusive.}
     \label{fig:fpo_633_10_20}
\end{figure}
We compare our uniform parameter optimization (UPO) with the fixed parameter optimization (FPO) technique suggested by Lee \emph{et al}. \cite{fixed_parameter} for obtaining good parameters. 
We apply it to find the least path from the trellis of a $[6,3,3]$ code upon receiving the erroneous vector $y = 111011$. 
Further, we analyze the performance of the two strategies by tracing out the occurrence of the actual solution state which is $\ket{011011}$, out of $2000$ simulation shots for $10$ training simulations. 
We observe that in Fig. \ref{fig:upo vs fpo 1} at each simulation instance after training the PQC with the UPO method, the actual solution has strikingly high occurrences. 
This shows the effectiveness of the proposed strategy towards obtaining good parameters. 
The fixed parameter optimization is quite random in obtaining a good set of optimized parameters.
We observe that there are instances where the fixed parameter optimization fails to obtain good parameters. 
In Fig. \ref{fig:fpo_633_10}, $3^\mathrm{rd}$, $5^\mathrm{th}$, $6^\mathrm{th}$ and $9^\mathrm{th}$ instance of the training shows the typical failures. 
Also, in Fig. \ref{fig:fpo_633_10_20} we observe FPO for $q = 20$ random initialization. The amplification of the solution state is still incomparable with the UPO.
We show further a success rate comparison of UPO and FPO in the Appendix \ref{ap:appendix4}.\\
The cost function of Eq. \ref{eq:cost_func} is essentially the expectation value of the observable in Eq. \ref{eq:cost_hamiltonian}. 
The state over which we calculate the expectation value is the state prepared by the PQC. 
Thus, the cost function to be minimized is
\begin{align}
    f(\psi) = \braket{\mathbf{0}|{U^\dagger_g} {U^\dagger_f}^p {U^\dagger_m}^p H_f U^p_m U^p_f U_g|\mathbf{0}},
\end{align}
where, $\ket{\mathbf{0}}$ is the all-zero state and $\psi$ refers to string representation of the state $\ket{\psi}$ prepared by the PQC, i.e.
\begin{align}
    \label{eq:land_state}
    \ket{\psi} = \ket{\psi(\beta,\gamma)} = U^p_m(\beta_o) U^p_f(\gamma_o) U_g \ket{\mathbf{0}}.
\end{align}
\begin{figure}[hbt!]
     \centering
     \begin{subfigure}[b]{0.49\textwidth}
         \centering
         \includegraphics[width = \textwidth]{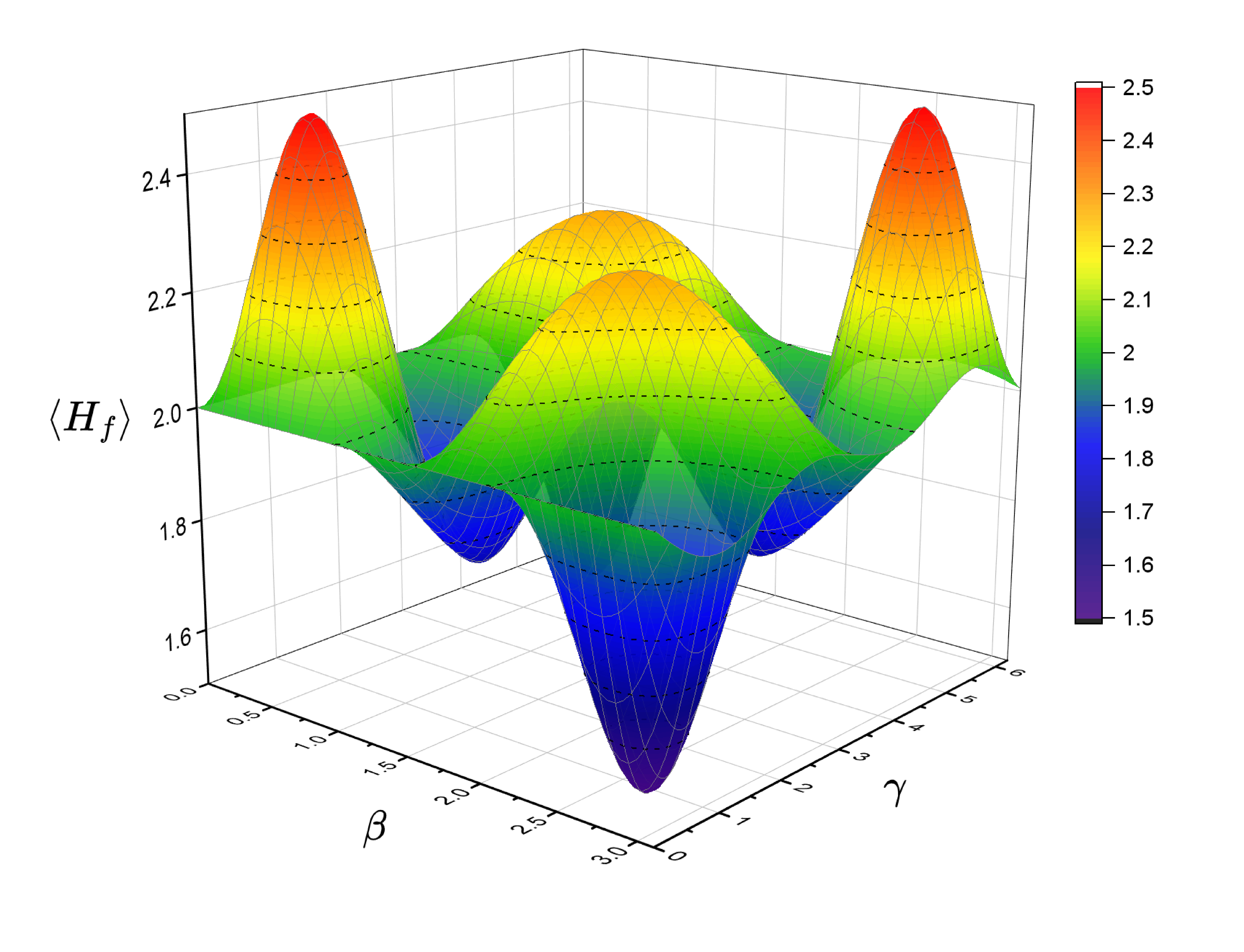}
         \caption{Cost function landscape for the $[3,2,1]$ code with $p = 3$.}
         \label{fig:land_321_3}
     \end{subfigure}
      \hfill
     \begin{subfigure}[b]{0.49\textwidth}
         \centering
          \includegraphics[width = \textwidth]{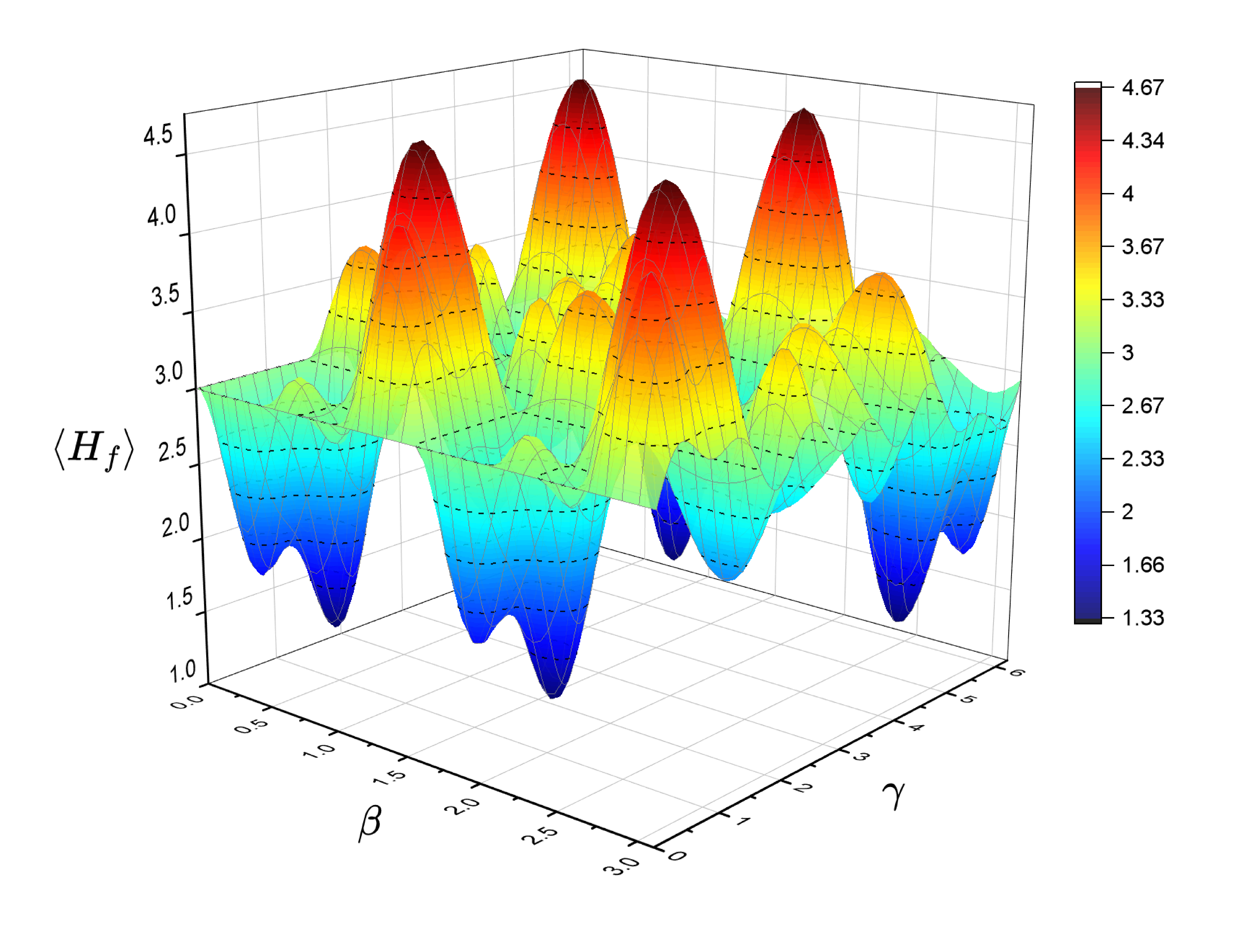}
         \caption{Cost function landscape for the $[6,3,3]$ code with $p = 3$.}
         \label{fig:land_633_3}
     \end{subfigure}

        \caption{It is a landscape of the expectation value of the observable cost Hamiltonian $H_f$ defined in Eq. \ref{eq:cost_hamiltonian}, over the state prepared by the UPO trained PQC Eq. \ref{eq:land_state}. The colour map signifies different surface heights for different colors. For both of the cases we observe no flat landscape.}
        \label{fig:lansdscape_p3}
\end{figure}
\begin{figure}[h!]
     \begin{subfigure}[b]{0.49\textwidth}
         \centering
         \includegraphics[width = \textwidth]{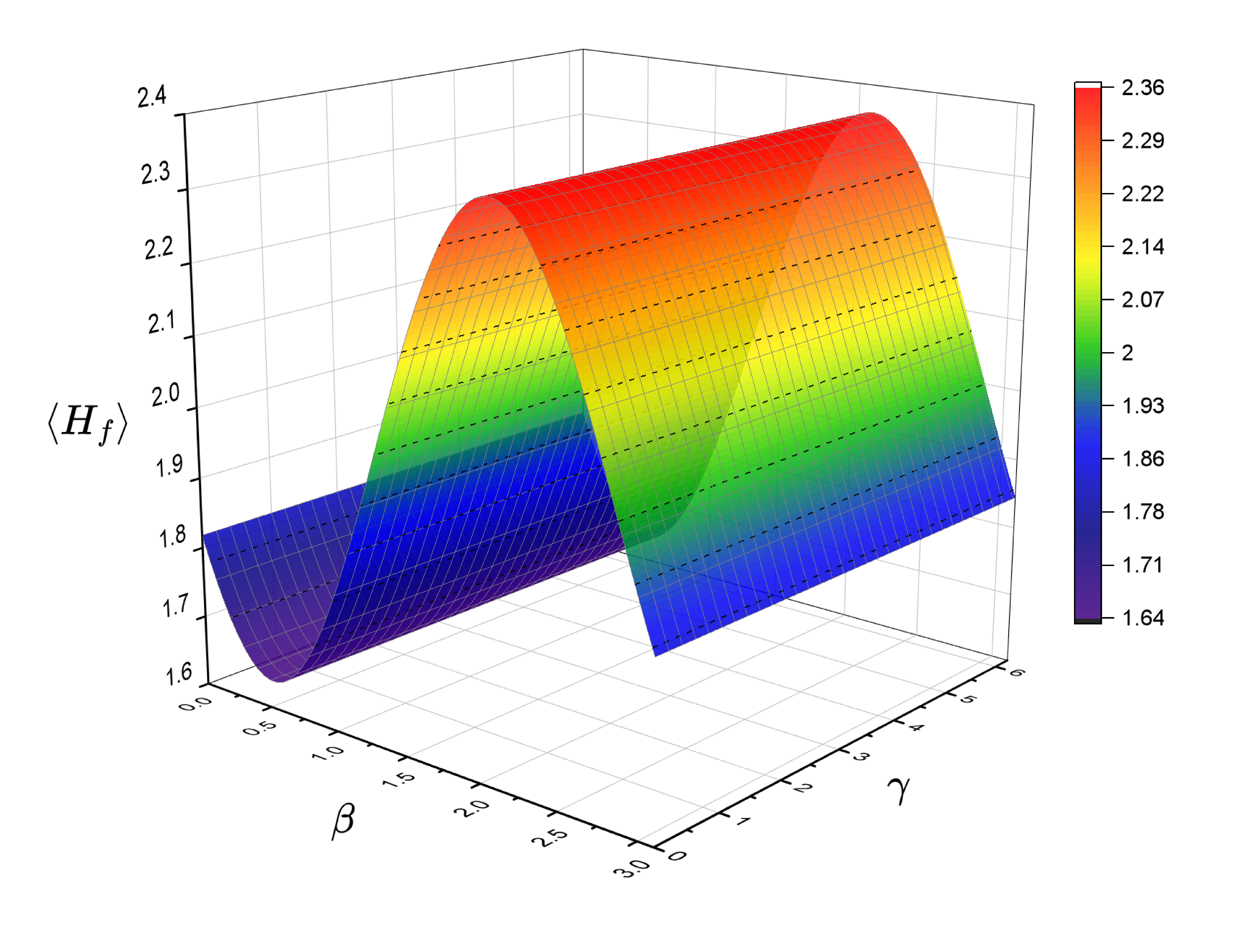}
         \caption{Cost function landscape for the $[3,2,1]$ code with $p = 3$.}
         \label{fig:land_321_fpo}
     \end{subfigure}
     \hfill
     \begin{subfigure}[b]{0.49\textwidth}
        \centering
        \includegraphics[width = \textwidth]{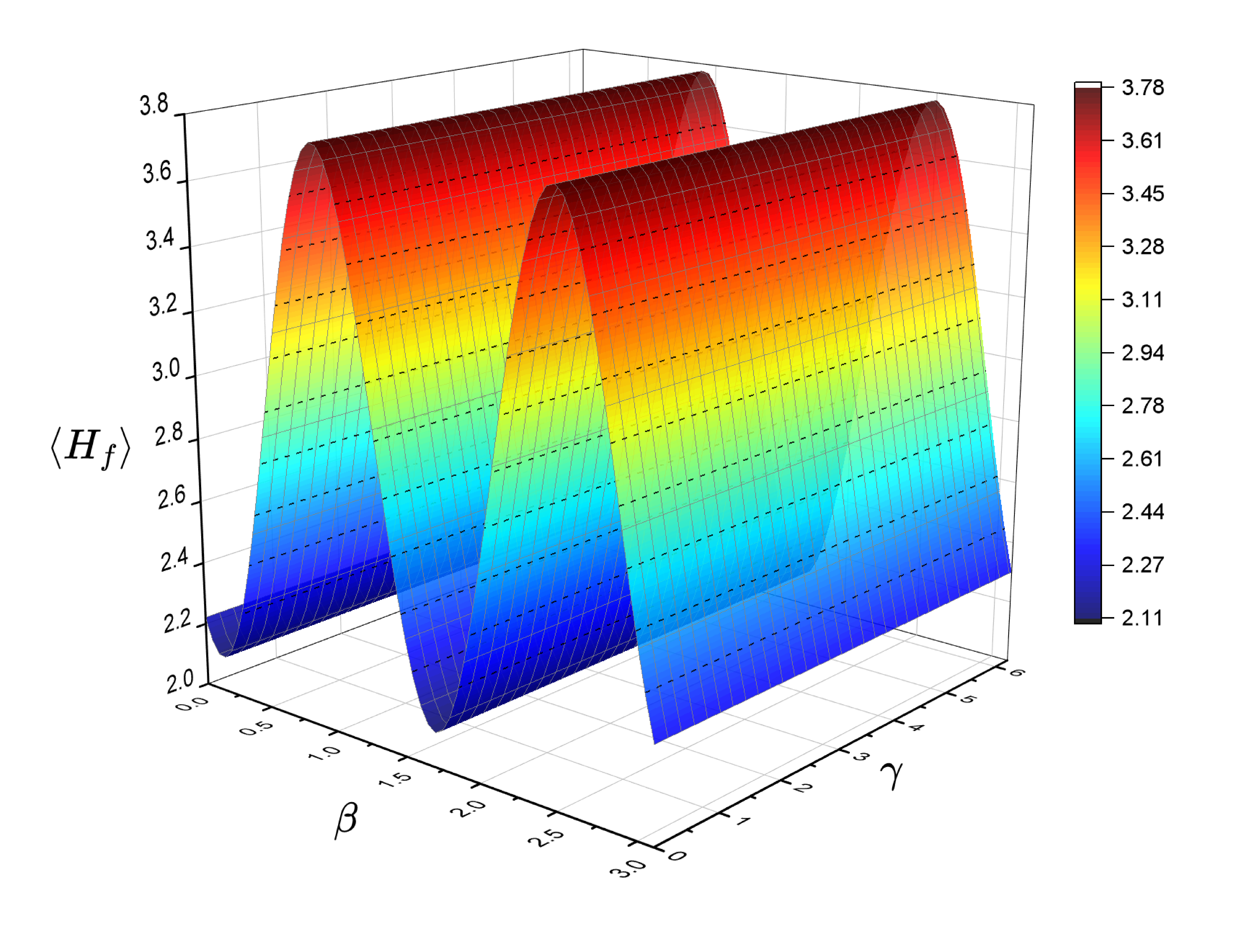}
         \caption{Cost function landscape for the $[6,3,3]$ code with $p = 3$.}
         \label{fig:land_633_fpo}
         \end{subfigure}
      \caption{A typical approximation of the cost function landscape for the $[6,3,3]$ and $[3,2,1]$ code with $p = 3$. The PQC is trained using FPO. We optimize $\beta_1,\beta_2$ and $\gamma_1, \gamma_2$ using FPO. Fixing them in those optimized values, we observe the landscape due to the $\beta_3$ and $\gamma_3$. Along gamma as we observe strips of zero gradients. These are shown by the strips of single colours along $\gamma_3$. It signifies $\braket{\frac{\partial f}{\partial \gamma_3}} = 0$.}
     \label{fig:fixed_land_633}
\end{figure}
In Fig. \ref{fig:lansdscape_p3}, we show the landscape of this cost function for $[3,2,1]$ and $[6,3,3]$ code, which results from the uniformly parameterized quantum circuit.
The optimization landscape is not flat and thus offers suitable optimization. 
This non-flat landscape explains the better efficiency of the UPO strategy. 
For the FPO, due to random initialization at each layer of unitaries the landscapes are expected to show more flat region. 
We visualize how barren plateaus still occur in the proposed hybrid Viterbi decoder if fixed parameter optimization is used.
For this we chose to work with the $[6,3,3]$ code. 
As the $6$ variable function can not be visualized with ease, we try to fix the first 4 optimized variables and plot the landscape for $\beta_3$ and $\gamma_3$. 
This help us to analyse the occurrence of vanishing gradients due to $\beta_3$ and $\gamma_3$. 
We observe in Fig. \ref{fig:fixed_land_633}, for a fixed $\beta_3$ there are strips of zero gradients arising due to the parameter $\gamma_3$.
Thus, fixing the parameters through a layerwise optimization offers no improvement in efficient optimization than random sampling.
We perform another set of experiments to test the decoders for a higher depth instance of $p = 15$ in Fig. \ref{fig:upo vs fpo 2}.
We now show the decoding results for other linear block codes and verify the multiple path tracing, all of which has least path metric.
\begin{figure}[ht!]
     \centering
     \begin{subfigure}[b]{0.48\textwidth}
         \centering
         \includegraphics[width = \textwidth]{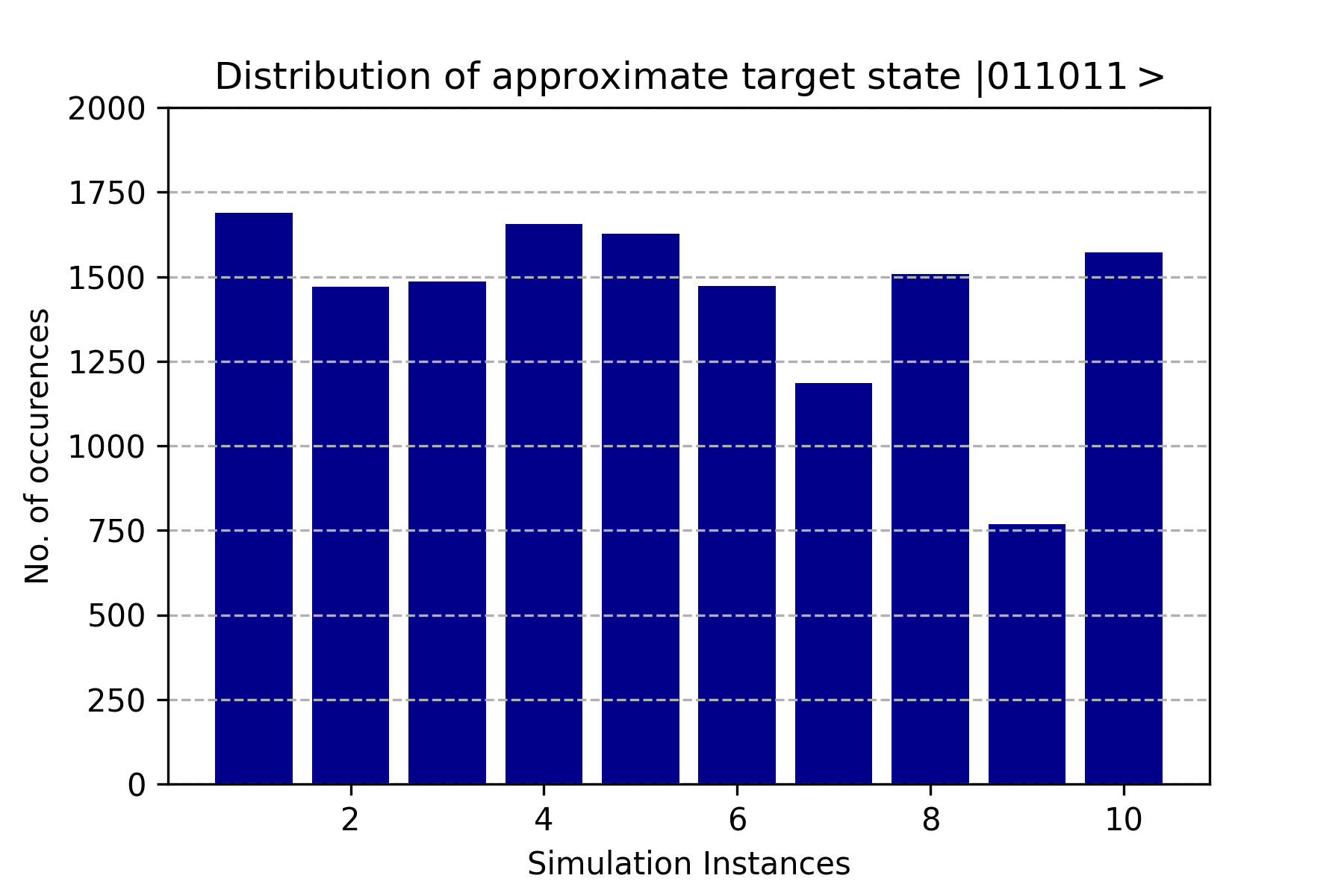}
         \caption{UPO}
         \label{fig: upo_633_depth_15}
     \end{subfigure}
      \hfill
     \begin{subfigure}[b]{0.48\textwidth}
         \centering
         \includegraphics[width = \textwidth]{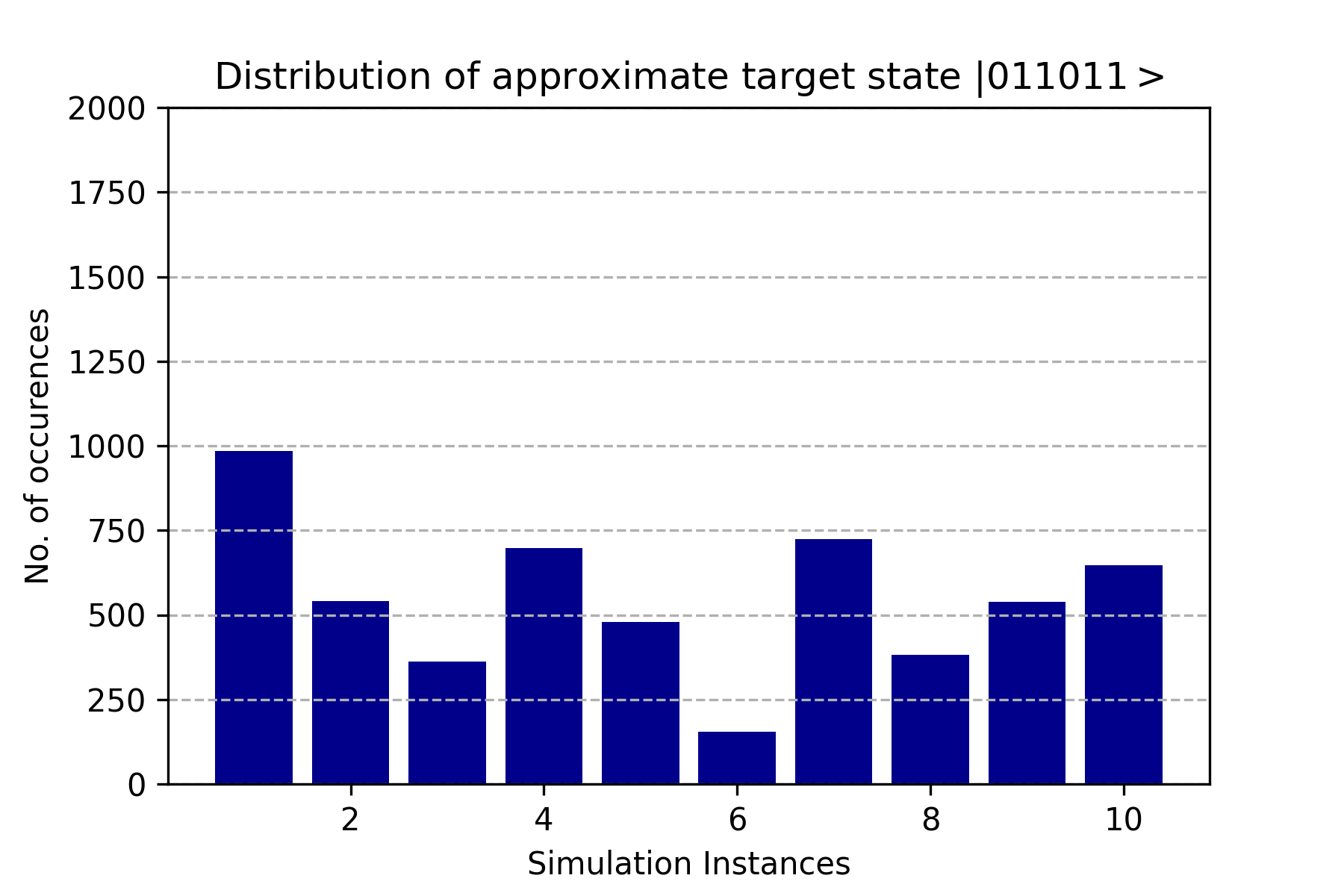}
         \caption{FPO}
         \label{fig:fpo_633_depth_15}
     \end{subfigure}

        \caption{Comparison between Uniform and fixed optimization strategy for a high depth PQC. 
        The number of unitary repetitions is $15$ in this PQC. 
        We have shown $10$ optimization instances and traced out the probability of occurrence  for the actual state with minimum path metric. 
        We choose $\mathcal{C} = [6,3,3], p = 15, q = 5$.}
        \label{fig:upo vs fpo 2}
\end{figure}
\subsection{Performance of the Hybrid Decoder}
\label{sec:Performance of the Hybrid Decoder}
We discuss the uniform parameter optimization technique to determine good parameters for the parameterized quantum circuit. 
After obtaining the parameters using UPO, we put them in the corresponding unitary layers and produce $\ket{\psi(\gamma, \beta)}$ of Eq. \ref{eq:param_unitary}. 
The operation of the hybrid Viterbi decoder is summarized in the following 3 steps:
\begin{enumerate}
    \item Prepare an initial superposition of all the valid paths in the trellis using a generator unitary $U_g$, same as in Eq. \ref{eq:in_state}.
    \item Prepare a parameterized quantum circuit having fixed depth $p$, using $U_f(\gamma_o)$ and $U_m(\beta_o)$, $p$ times
    and train it using uniform parameter optimization strategy.
    \item Use the parameters obtained from the training of the PQC and measure the qubits.
\end{enumerate}
\begin{figure}[hbt!]
     \centering
     \begin{subfigure}[b]{0.49\textwidth}
         \centering
         \includegraphics[width = \textwidth]{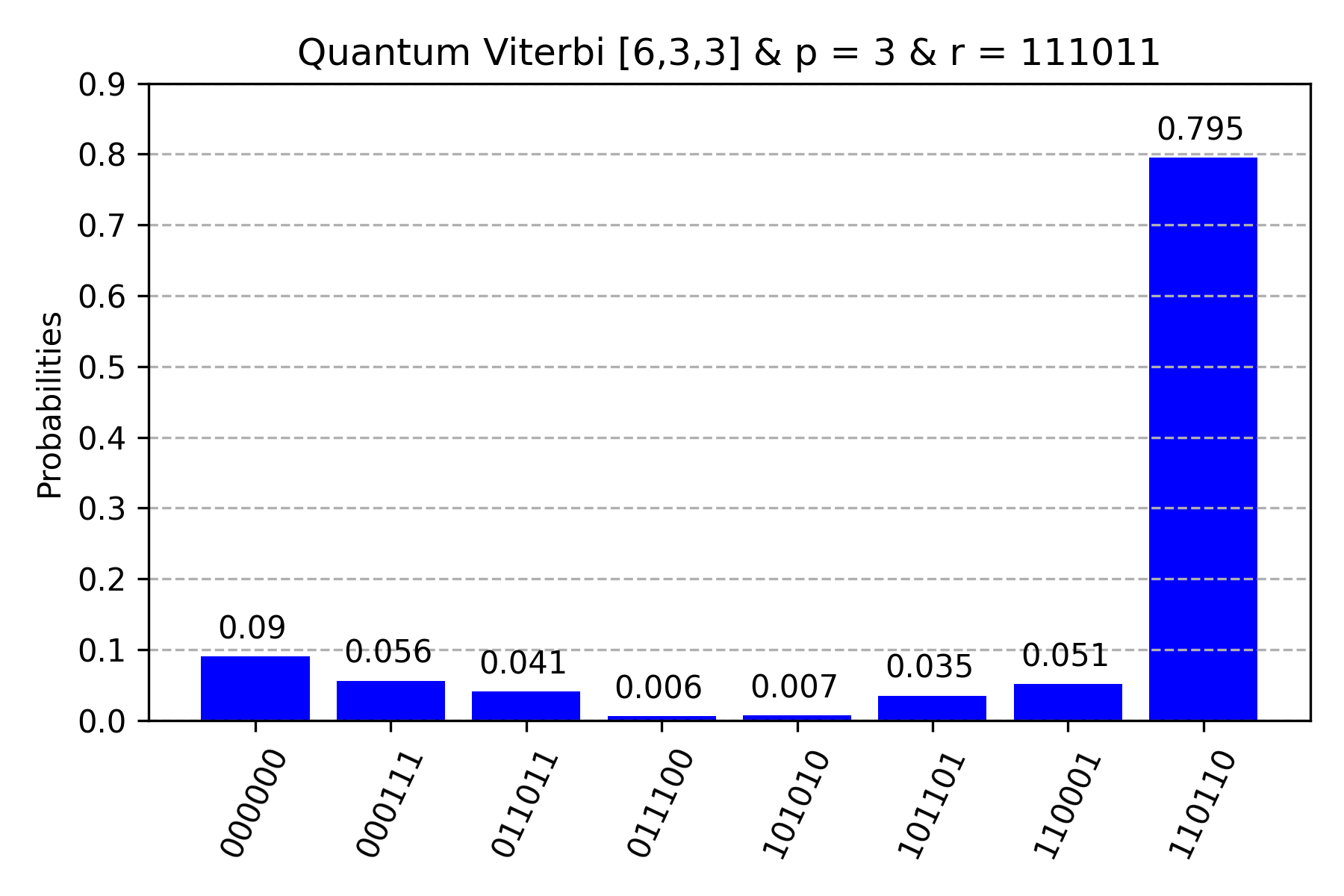}
         \caption{Solution found using UPO.}
         \label{fig: sol_633_uni_1}
     \end{subfigure}
      \hfill
     \begin{subfigure}[b]{0.49\textwidth}
         \centering
         \includegraphics[width = \textwidth]{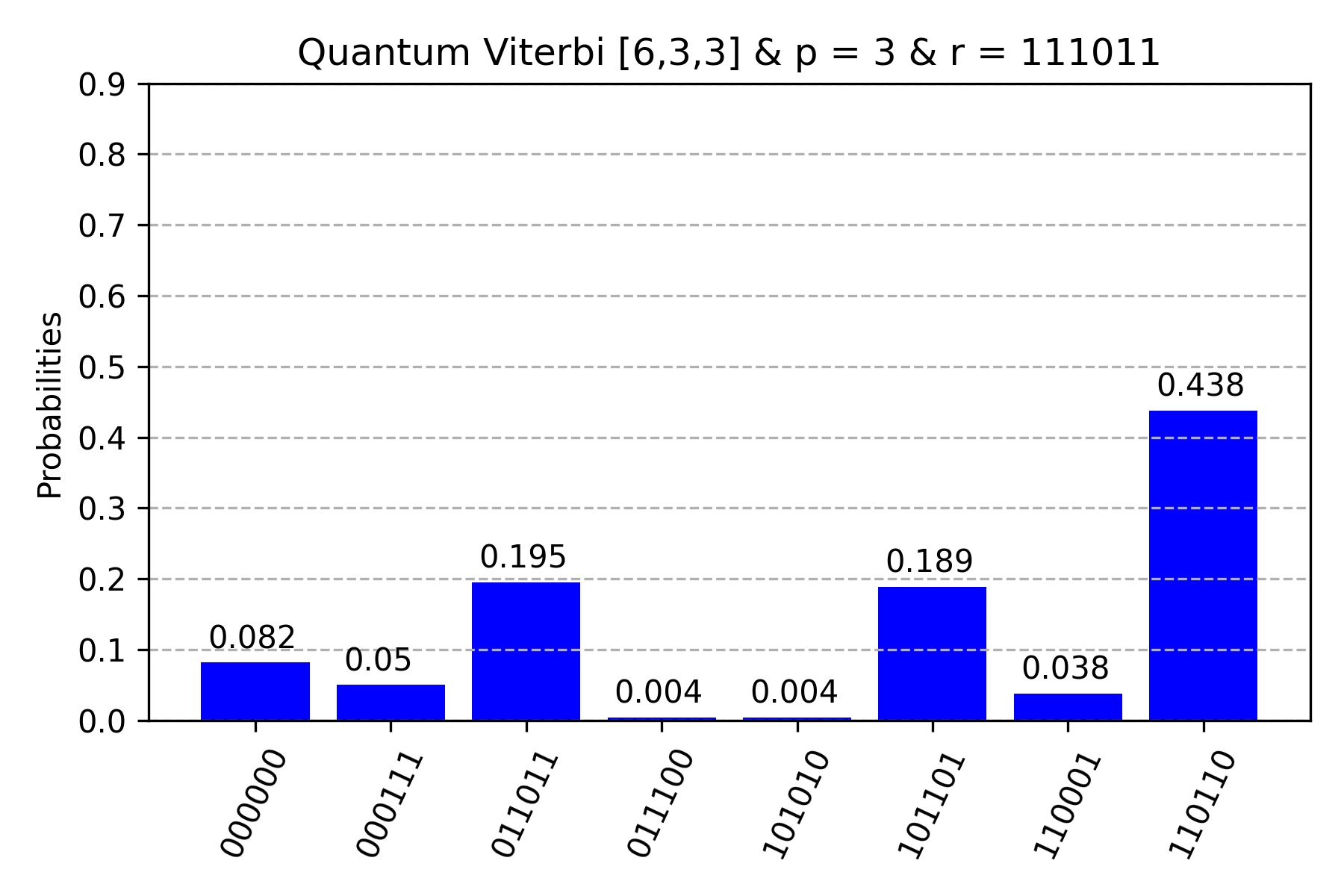}
         \caption{Solution found using FPO.}
         \label{fig:sol_633_fix_1}
     \end{subfigure}

        \caption{Results of the simulation obtained for Viterbi decoding by the QAOA using UPO. The code is $[6,3,3]$. We label the states with convention that the most significant bit is at the bottom. We apply $p = 3$ number of unitary layers. We sample out of $q = 5$ randomly initialized samples. Received vector is $y = 111011$. The amplified state refers to the decoded codeword or the path with the least metric in the trellis. For a fixed $p$ and $q$, UPO offers better amplification than FPO.}
        \label{fig:upo vs fpo soln}
\end{figure}
In Fig. \ref{fig:upo vs fpo soln}, we show the solution obtained with high probability for a $[6,3,3]$ code upon receiving an erroneous vector $y=111011$. 
The state $\ket{011011}$ has the highest probability of occurrence and thus it is the solution we obtain using QAOA. 
The decoded path has a path metric $d(011011,y) = 1$.
We note that this is indeed the path with minimum path metric corresponding to the received erroneous vector. 
We also show the amplification difference between the two optimization techniques, which again points out the more acceptance of the UPO compared to FPO.
We also apply the algorithm for a $[9,4,4]$ code upon receiving an erroneous vector $y = 111011011$.
The result is shown in Fig. \ref{fig: sol_944_uni_1}. 
Here, we again use UPO to train the parameterized quantum circuit. 
The PQC we use has 3 unitary layers.
The decoded codeword has unit path metric and is also the minimum among others.
In another case where we took a $[3,2,1]$ code. 
Codespace of this code is $\mathcal{C} = \{000, 010, 101, 111\}$. 
For this code, if we receive the erroneous vector $y = 011$, the algorithm returns two states with high probability $\ket{010}$ and $\ket{111}$ as shown in Fig. \ref{fig:sol_32_uni_1}. 
This suggests there are two paths with the least path metric, and this indeed matches the theory. 
It signifies that every path which has the minimum metric among all the paths will be opted out as the solution by the hybrid decoder with equal probability. 
\begin{figure}[hbt!]
     \centering
      \begin{subfigure}[t]{0.49\textwidth}
         \centering
         \includegraphics[width =\textwidth]{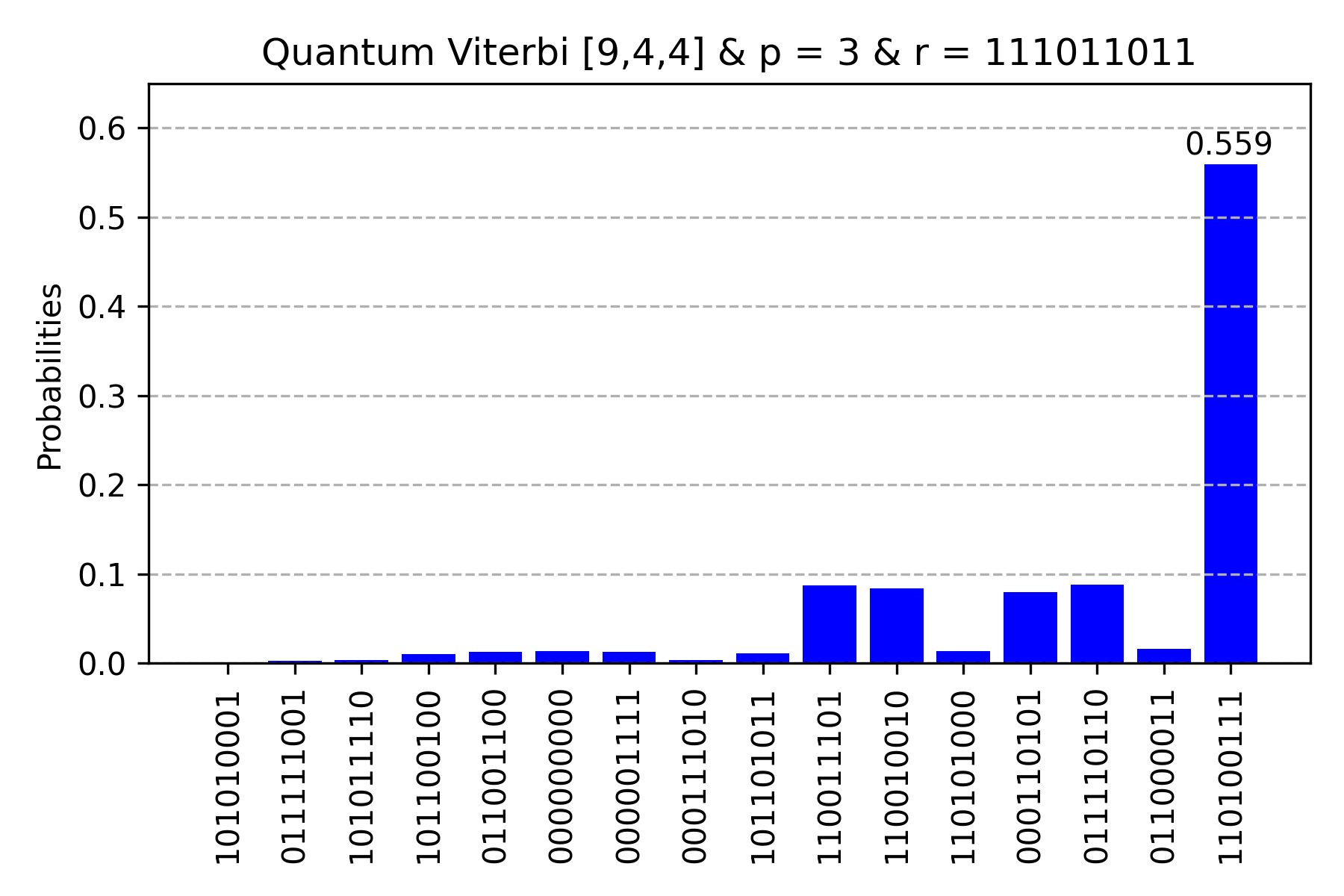}
         \caption{$[9,4,4]$ code}
         \label{fig: sol_944_uni_1}
         \end{subfigure}
      \hfill
      \begin{subfigure}[t]{0.49\textwidth}
         \centering
         \includegraphics[width =\textwidth]{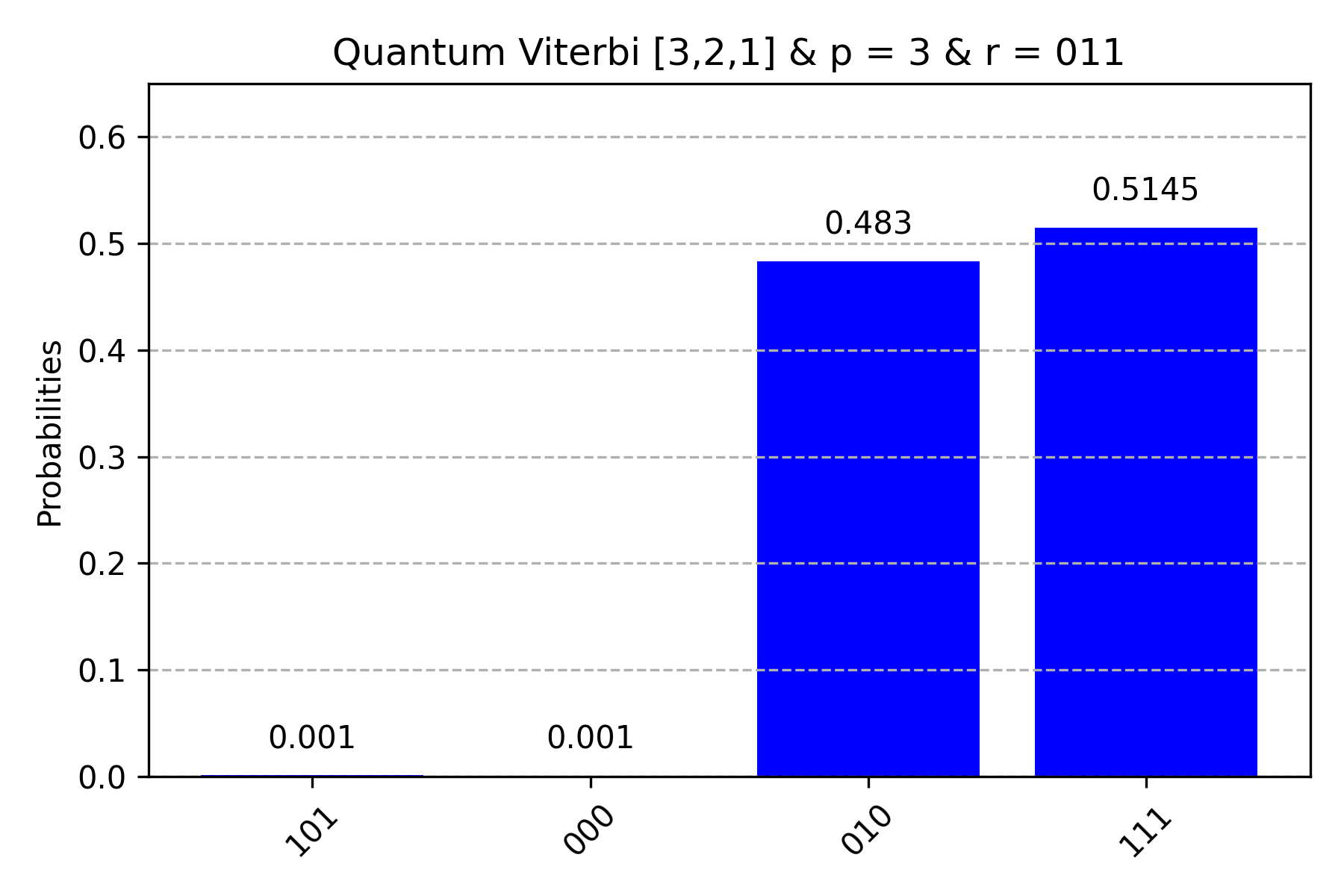}
          \caption{$[3,2,1]$ code}
         \label{fig:sol_32_uni_1}

        \label{fig: upo soln}
        \end{subfigure}
        \label{fig:misc}
        \caption{We use UPO for the hybrid decoder of a $[9,4,4]$ code in (a) and $[3,2,1]$ code in (b).
        UPO uses $p = 3$ ,$q = 5$.
        We label the states with convention that the most significant bit is at the bottom.
        For the $[9,4,4]$ code, the decoder estimates the path $111001011$ having minimum path metric; corresponding to received $y = 111011011$. The path metric is $d(111001011,y) = 1$.
        For the $[3,2,1]$ code, the decoder estimates two paths having the least path metric corresponding to the received $y = 011$. 
        The path metric is $d(010,y) = d(111,y) = 1$.}
\end{figure}
\begin{figure}
    \centering
    \includegraphics[width=0.65\textwidth]{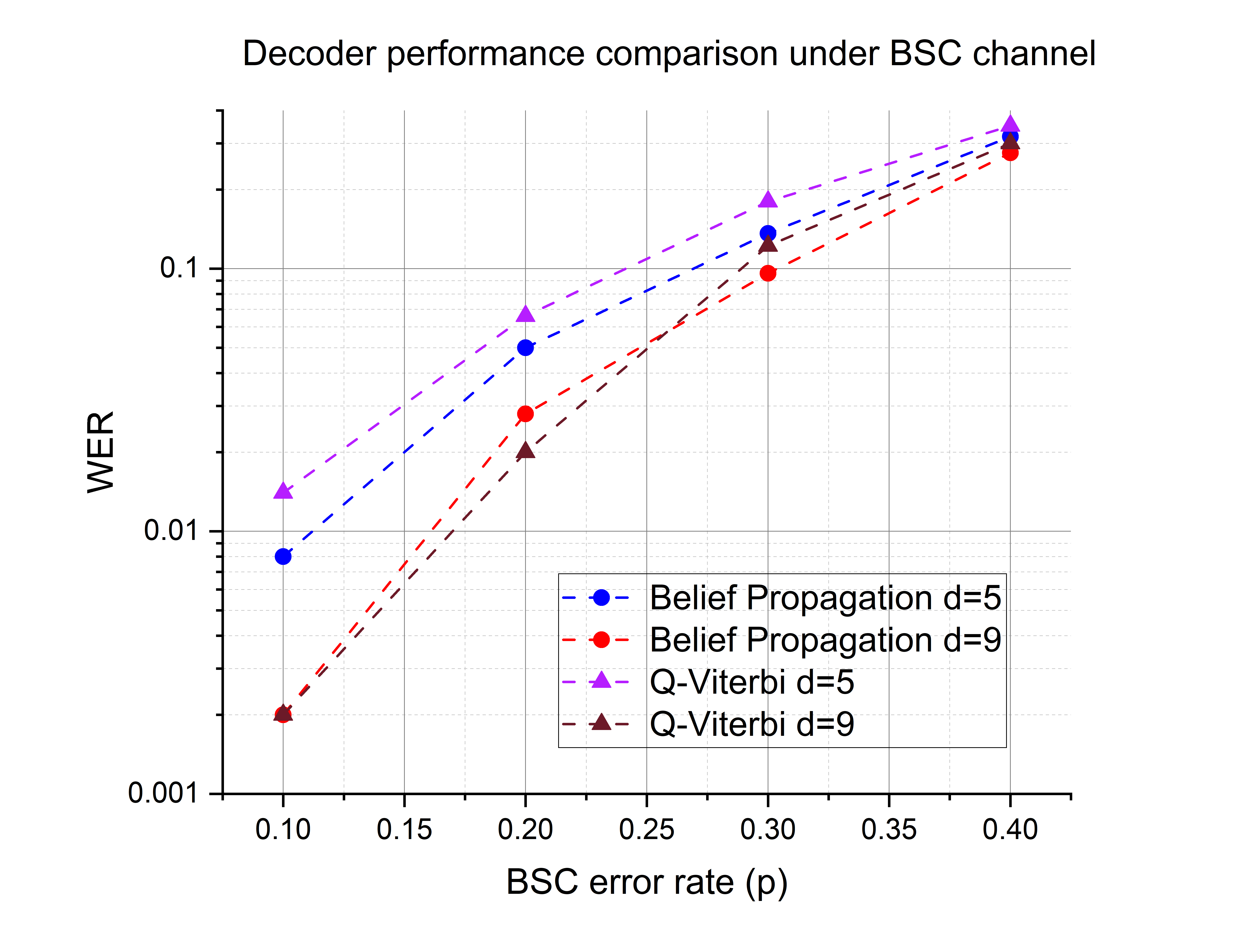}
    \caption{We compare the performance of the hybrid decoder (here referred as Q-Viterbi in short) with the belief propagation decoder.
    We benchmark the word error rates (WER) of repetition codes under the binary symmetric channel (BSC). 
    We observe an almost comparable performance of both the decoders, where the slight difference in the word error rates are due to statistical errors appeared as a result of the limited number of sampling instances.}
    \label{fig:wer-rep-code}
\end{figure}
We further benchmark the performance of the proposed hybrid decoder with standard \textit{min-sum} belief propagation decoder.
We use classical repetition code, which has only two possible codewords $\mathcal{C} = \{0^d, 1^d\}$ for a distance $d$ code.
According to theorem \ref{theorem_mixer}, the mixer Hamiltonian has only one term i.e. $U_m = \mathrm{X}_1 \mathrm{X}_2 .... \mathrm{X}_d$.
The parity check matrix of a repetition code has no cycles and therefore belief propagation decoder converges to exact error pattern, which satisfy the syndrome.
Also, as our hybrid decoder is essentially finding the minimum Hamming weight codeword,
the hybrid decoder also should predict every correctable error patterns satisfying the syndrome.
Therefore, we benchmark the word error rates obtained for both the decoders under various bit flip probability of a binary symmetric channel.
We observe a very similar performance of the classical and hybrid decoder.
This warrant the applicability of the hybrid decoder as an active decoder for classical error correction.
\section{Conclusion and future research}
We proposed a Quantum Approximation Optimization algorithm, which uses parameters uniformly optimized for the mixer and cost unitaries in the PQC. 
We showed through simulations that taking the parameters uniformly helps better optimization than taking them non-uniformly.
We used this variational approach to construct a hybrid Viterbi decoder.
We observed that the constructed hybrid decoder opts out any number of shortest paths present in the trellis as a solution. This quantum-classical decoder shares the workload and it even works well for shallow unitary depths.
We showed that the structure of the unitaries depend on the type of classical error-correcting code.
Depending on the classical code the hybrid decoder is dealing with, the complexity of both the cost and mixer Hamiltonian might increase.
For instance, the locality and the structure of mixer $U_m$ depends on the variable `$v$', as discussed in Theorem \ref{theorem_mixer}.
These will reflect on the circuit depth of the PQC and might lead to a hindrance for it's near-term implementation.
Thus, the hybrid decoder in general will require assistance from quantum error-correcting codes to be implemented in NISQ devices such that large classical decoding problems can be addressed.\\
We point out future work that can be extended on this idea.
Regarding the near-term implementation, we propose a possible connection with the work related to a modified variational quantum algorithm \cite{zhang2022variational}.
On a similar note, it would be valuable to test the QAOA-based Viterbi decoder on larger codes to evaluate the performance of UPO. 
For small codes, we demonstrate that UPO outperforms both the existing random initialization and FPO. 
For large codes, where the vast number of parameters renders QAOA-based decoding intractable, our proposed method is naturally expected to perform better than the two other alternatives. 
Currently, we are constrained by the resources available in the \textit{qiskit qasm simulator} and thus we reserve the scalability experiment of the hybrid decoder as a near future work.
Further, finding the shortest path using a variational approach can be extended beyond classical encoding-decoding problems. Inference problems can also be approached using the method we adopted. 
One can design a unitary to produce an equal superposition of all the possible solutions. 
Depending on the memory of the system, one can try to minimize certain path metrics related to the problem to suggest a suitable solution.
Although for search problems Grover's search is optimal, this algorithm might be an alternative in terms of circuit depth. The quantum-classical hybrid approach will generally require low-depth quantum circuits compared to the Grover's search.

\section{Data availability}
All the scripts used for generating the numerical results are available from the corresponding author upon reasonable request.


\bibliography{mybibliography}

\appendix
\section{Appendix}
\subsection{Mapping of a Hermitian and it's time evolution unitary}
\label{ap:appendix1}
We elaborate one of the steps we used to proof Theorem \ref{theorem_mixer}. In Eq. \ref{eq:span of mixer}, we implied \begin{align*}
        e^{-i\beta H_m}\ket{\psi_{in}}  \in \mathcal{C}
        \implies & H_m\ket{\psi_{in}}  \in \mathcal{C}.
\end{align*}
Now as $H_f$ and $H_m$ are non commuting, $\ket{\psi_{in}}$ is not an eigenstate of $H_m$. 
Although, the eigenstates of the Hermitian $H_m$ forms a basis and therefore $\ket{\psi_{in}}$ can be expressed as a linear combination of them.
\begin{align}
    \ket{\psi_{in}} = \sum_{i=1}^{n}a_i\ket{E_m}_i,
    \label{eq:a_psiin}
\end{align}
where $\ket{E_m}_i$ are the eigenstates of $H_m$ and $a_i$ are the coefficients. 
Also $\ket{\psi_{in}} \in \mathcal{C}$.
Now the operator $U_m = e^{-i\beta H_m}$ essentially is the time evolution operator, with $\beta$ as the time parameter. We can express the state after the operation of $U_m$ as:
\begin{align}
    \psi(\beta) = U_m\ket{\psi_{in}}
                = e^{-i\beta H_m}\ket{\psi_{in}}.
    \label{eq:a_psibeta}
\end{align}
Any exponential of a matrix can be expressed as a power series expansion,
\begin{align}
    e^{-i\beta H_m} = \mathrm{I} - i\beta H_m - \frac{{\beta}^2}{2!} H^2_m + .....
    \label{eq:power}
\end{align}
This power series converges for certain power of the Hamiltonian. 
Now as $\ket{\psi_{in}}$ is not an eigenstate of $H_m$. Eq. \ref{eq:a_psibeta} cannot be expressed in a closed form. But $\ket{E_m}_i$ are eigenstates of the Hamiltonian $H_m$ thus,
\begin{align}
    U_m\ket{E_m}_i &= e^{-i\beta H_m}\ket{E_m}_i,\nonumber \\
                   &= (\mathrm{I} - i\beta H_m - \frac{{\beta}^2}{2!} H^2_m + .....)\ket{E_m}_i, \nonumber \\
                   &= (\mathrm{I} - i\beta {E_m}_i - \frac{{\beta}^2}{2!} {E^2_m}_i + .....)\ket{E_m}_i, \nonumber \\
                   &= e^{-i\beta {E_m}_i}\ket{E_m}_i,
    \label{eq:ue}
\end{align}
where ${E_m}_i$ is the eigenvalue of $H_m$ corresponding to the eigenstate $\ket{E_m}_i$.
Now we can use all these tools to show why both the time evolution unitary and the corresponding hermitian has the same mapping. According to Theorem \ref{theorem_mixer}, $H_m\ket{\psi_{in}} \in \mathcal{C}$. 
As the linear block code is linear, following Eq. \ref{eq:a_psiin} it is obvious that ${E_m}_i \in \mathcal{C}$. Now using Eq. \ref{eq:ue} and completeness of the eigenstates of $H_m$ in Eq. \ref{eq:a_psibeta}, we have
\begin{align}
    \psi(\beta) &= e^{-i\beta H_m}\ket{\psi_{in}}, \nonumber \\
                &= \sum_{i}e^{-i\beta H_m}\ket{E_m}_i\bra{E_m}_i\ket{\psi_{in}}, \nonumber \\
                &= \sum_{i}e^{-i\beta {E_m}_i}\ket{E_m}_i\bra{E_m}_i\ket{\psi_{in}}, \nonumber \\
                & = \sum_{i}a_i e^{-i\beta {E_m}_i}\ket{E_m}_i.
    \label{eq:append_last}
\end{align}
Eq. \ref{eq:append_last} proves that $\psi(\beta) \in \mathcal{C}$. Expressing the initial state as a superposition of the eigenstates $H_m$ only adds phases into the corresponding terms in the superposition when $U_m$ is applied to it. Thus, the unitary $U_m$ and the Hamiltonian $H_m$ have the same mapping on states. 
\subsection{The cost Hamiltonian}
\label{ap:appendixproof}
Here we show the derivation for the cost Hamiltonian $H_f$.
Consider a Boolean function $ C(c): {\{-1,1\}}^n \to \{-1,1\}$. The Fourier expansion of the function is \cite{boolean_func}
\begin{align}
    \label{eq:furie_r}
    C(c) & = \sum_{S\subset [n]} \hat{C}(s)\prod_{i \in S}c_i, \\
        & = \sum_{S\subset [n]} \hat{C}(s) \chi_{S},
\end{align}
where $\chi_{S} = \prod_{i \in S}c_i$, is the parity function and can be mapped to Pauli-$\mathrm{Z}$ operators of the form $\prod_{i \in S} \mathrm{Z}_i$
and $[n] = \{-1,1\}^n$.
This mapping is obvious when we work in the computational basis states, since $\ket{0}$ and $\ket{1}$ are the eigenstates of Pauli-$\mathrm{Z}$ with eigenvalues $+1$ and $-1$ respectively.
Thus, the Hamiltonian corresponding to the Boolean function $C(c)$ is
\begin{align}
    H_C = \sum_{S\subset [n]} \hat{C}(s) \prod_{i \in S} \mathrm{Z}_i,
    \label{eq:hamiltonian_map}
\end{align}
where $\hat{C}(s)$ are the Fourier coefficients and solely depend on the value of the function. 
Eq. \ref{eq:furie_r} is a real multilinear polynomial, which is further simplified as
\begin{align}
    \label{eq:fourier}
    C(c) & = \sum_{a\in \{-1,1\}^n} C(a) \mathbb{I}(c), \nonumber \\
         & = \sum_{a\in \{-1,1\}^n}C(a)\left(\frac{1+a_1c_1}{2}\right)\left(\frac{1+a_2c_2}{2}\right) \cdots \left(\frac{1+a_nc_n}{2}\right),
\end{align}
where the sum runs over all possible combinations of $a$.
If $C: \{-1,1\}^2 \to \{-1,1\}$, we define one such function as 
\begin{align}
C(c) = c_1 \oplus y_1,
\label{eq:boolean_clause}
\end{align}
where $c = (c_1,y_1)$.
Now Eq. \ref{eq:fourier} becomes,
\begin{align}
    \label{eq:fu_ex}
    C(c) =& \,C(-1,-1)\left(\frac{1-c_1}{2}\right)\left(\frac{1-y_1}{2}\right) + C(-1,1)\left(\frac{1-c_1}{2}\right)\left(\frac{1+y_1}{2}\right)\nonumber \\ & + C(1,-1)\left(\frac{1+c_1}{2}\right)\left(\frac{1-y_1}{2}\right) + C(1,1)\left(\frac{1+c_1}{2}\right)\left(\frac{1+y_1}{2}\right).
\end{align}
Putting the value of $C(c)$ as defined in Eq. \ref{eq:boolean_clause}, the Fourier expansion of Eq. \ref{eq:fu_ex} becomes
\begin{align}
    \label{eq:cx}
    C(c) =&\, - 1\left(\frac{1-c_1}{2}\right)\left(\frac{1-y_1}{2}\right) + 1\left(\frac{1-c_1}{2}\right)\left(\frac{1+y_1}{2}\right) \nonumber \\ 
    & + 1\left(\frac{1+c_1}{2}\right)\left(\frac{1-y_1}{2}\right) - 1\left(\frac{1+c_1}{2}\right)\left(\frac{1+y_1}{2}\right), \nonumber \\
    & = -c_1y_1.
\end{align}
It can be easily verified from the above equation that the multilinear expansion satisfies the definition of $C(c)$, given in Eq. \ref{eq:boolean_clause}. 
If the range of the Boolean function is changed from $\{-1,1\}$ to \{0,1\}, the Fourier expansion of the function $C(c)$ derived in Eq. \ref{eq:cx} becomes \cite{boolean_func}
\begin{align}
    C^{\prime}(c) &= \frac{1 + C(c)}{2}, \nonumber \\
                & = \frac{1}{2}\left(1-c_1y_1\right).
    \label{eq:fun_c}
\end{align}
Now using the Hamiltonian mapping obtained in Eq. \ref{eq:hamiltonian_map}, we find the Hamiltonian for $C^{\prime}(c)$ becomes
\begin{align}
    \label{eq:ham_c}
    H_C = \frac{1}{2}\left(\mathrm{I}-\mathrm{Z}_c \mathrm{Z}_y\right).
\end{align}
This is also the cost Hamiltonian corresponding to the function $C(c)$, given in Eq. \ref{eq:boolean_clause},
because $C^{\prime}(c)$ is the Fourier expansion of $C(c)$.
We mention that the cost function for the Viterbi decoder $f(c)$ in Eq. \ref{eq:cost_func} is the weighted sum of $C(c)$ with all the weights are $1$. 
Thus, mapping the variables in Eq. \ref{eq:fun_c} to appropriate registers, i.e. $c_i \to x_i$ and $y_i \to r_i$ and then summing over $i$ from $1$ to $n$, Eq. \ref{eq:ham_c} becomes $H_c \to H_f = \displaystyle \sum_{i=1}^{n} \left(\frac{1}{2}\mathrm{I} - \frac{1}{2}{\mathrm{Z}_{x_i}}{\mathrm{Z}_{r_i}}\right)$.
This is the cost Hamiltonian of the Viterbi decoder we propose in Eq. \ref{eq:cost_hamiltonian}.
\newpage
\subsection{Circuit of the mixer for the [6,3,3] code}
\label{subsec:Circuit of the mixer}
The mixer for the $[6,3,3]$ linear block code is defined as $U_m(\beta) = e^{-i\beta H_m}$. $U_m(\beta)$ is expressed in Eq. \ref{eq:mixmix}. We show the circuit diagram of the mixer below.
\begin{figure}[h!]
\advance\leftskip-0.5cm
\begin{quantikz}[font = \tiny]
    \lstick{$\ket{x_1}$} & \gate{\mathrm{H}} & \ctrl{1} & \qw & \qw & \qw & \ctrl{1} & \gate{\mathrm{H}} & \gate{\mathrm{H}} & \ctrl{4} & \qw & \qw & \qw & \ctrl{4} & \gate{\mathrm{H}} & \qw \hdots \hdots\\
    \lstick{$\ket{x_2}$} & \gate{\mathrm{H}} & \targ{} & \ctrl{1} & \qw & \ctrl{1} & \targ{} & \gate{\mathrm{H}} & \qw & \qw & \qw & \qw & \qw & \qw & \qw & \qw \hdots \hdots\\
    \lstick{$\ket{x_3}$} & \gate{\mathrm{H}} & \qw & \targ{} & \gate{R_\mathrm{z}(2\beta)} & \targ{} & \qw & \gate{\mathrm{H}} & \qw & \qw & \qw & \qw & \qw & \qw & \qw & \qw \hdots \hdots\\
    \lstick{$\ket{x_4}$} & \qw & \qw & \qw & \qw & \qw & \qw & \qw & \qw & \qw & \qw & \qw & \qw & \qw & \qw & \qw \hdots \hdots\\
    \lstick{$\ket{x_5}$} & \qw & \qw & \qw & \qw & \qw & \qw & \qw & \gate{\mathrm{H}} & \targ{} & \ctrl{1} & \qw & \ctrl{1} & \targ{} & \gate{\mathrm{H}} & \qw \hdots \hdots\\
    \lstick{$\ket{x_6}$} &  \qw & \qw & \qw & \qw & \qw & \qw & \qw & \gate{\mathrm{H}} & \qw & \targ{} & \gate{R_\mathrm{z}(2\beta)} & \targ{} & \qw & \gate{\mathrm{H}} & \qw \hdots \hdots\\
    \lstick{$\ket{x_1}$} \hdots \hdots & \qw & \qw & \qw & \qw & \qw & \qw & \qw & \qw & \qw & \qw & \qw & \qw & \qw & \qw & \qw &\\
    \lstick{$\ket{x_2}$} \hdots \hdots & \qw & \qw & \qw & \qw & \qw & \qw & \qw & \gate{\mathrm{H}} & \ctrl{2} & \qw & \qw & \qw & \ctrl{2} & \gate{\mathrm{H}} & \qw &\\
    \lstick{$\ket{x_3}$} \hdots \hdots & \gate{\mathrm{H}} & \ctrl{1} & \qw & \qw & \qw & \ctrl{1} & \gate{\mathrm{H}} & \qw & \qw & \qw & \qw & \qw & \qw & \qw & \qw &\\
    \lstick{$\ket{x_4}$} \hdots \hdots & \gate{\mathrm{H}} & \targ{} & \ctrl{1} & \qw & \ctrl{1} & \targ{} & \gate{\mathrm{H}} & \gate{\mathrm{H}} & \targ{} & \ctrl{2} & \qw & \ctrl{2} & \targ{} & \gate{\mathrm{H}} & \qw &\\
    \lstick{$\ket{x_5}$} \hdots \hdots & \gate{\mathrm{H}} & \qw & \targ{} \qw & \gate{R_\mathrm{z}(2\beta)} & \targ{} & \qw & \gate{\mathrm{H}} & \qw & \qw & \qw & \qw & \qw & \qw & \qw & \qw &\\
    \lstick{$\ket{x_6}$} \hdots \hdots & \qw & \qw & \qw & \qw & \qw & \qw & \qw & \gate{\mathrm{H}} & \qw & \targ{} & \gate{R_\mathrm{z}(2\beta)} & \targ{} & \qw & \gate{\mathrm{H}} & \qw &
\end{quantikz}
    \caption{Mixer $U_m(\beta)$ for the $[6,3,3]$ linear block code. We exclude the ancilla register for showing the mixer circuit as they has no part in it. Horizontal dots are used to represent the circuit continuation. $R_z(-\gamma)$ is the single qubit rotation gate about the $z$ axis, $R_z(2\beta) = e^{-i \beta \mathrm{Z}}$.}
    \label{fig:appendix_mixer}
\end{figure}

\newpage

\subsection{Circuit of the cost unitary for the [6,3,3] code}
\label{subsec:Circuit of the cost unitary}
Here we show the detail circuit of how we implement the cost unitary $U_f(\gamma)$. The cost unitary for the $[6,3,3]$ code is defined in Eq.\ref{eq:cost_unitary}. In Fig.\ref{fig:appendix_cost}, the connection between the codeword register $x$ and the ancilla register $r$ is shown.
\begin{figure}[h!]
\advance\leftskip-0.5cm
\begin{quantikz}[font = \tiny]
\lstick{$\ket{x_1}$} & \ctrl{6} & \qw & \ctrl{6} & \qw & \qw & \qw & \qw & \qw & \qw & \qw & \qw & \qw & \qw \hdots \hdots\\
\lstick{$\ket{x_2}$} & \qw & \qw & \qw & \ctrl{6} & \qw & \ctrl{6} & \qw & \qw & \qw & \qw & \qw & \qw & \qw \hdots \hdots\\
\lstick{$\ket{x_3}$} & \qw & \qw & \qw & \qw & \qw & \qw & \ctrl{6} & \qw & \ctrl{6} & \qw & \qw & \qw & \qw \hdots \hdots\\
\lstick{$\ket{x_4}$} & \qw & \qw & \qw & \qw & \qw & \qw & \qw & \qw & \qw & \ctrl{6} & \qw & \ctrl{6} & \qw \hdots \hdots\\
\lstick{$\ket{x_5}$} & \qw & \qw & \qw & \qw & \qw & \qw & \qw & \qw & \qw & \qw & \qw & \qw & \qw \hdots \hdots\\
\lstick{$\ket{x_6}$} & \qw & \qw & \qw & \qw & \qw & \qw & \qw & \qw & \qw & \qw & \qw & \qw & \qw \hdots \hdots\\
\lstick{$\ket{r_1}$} & \targ{} & \gate{R_\mathrm{z}(-\gamma)} & \targ{} & \qw & \qw & \qw & \qw & \qw & \qw & \qw & \qw & \qw & \qw \hdots \hdots\\
\lstick{$\ket{r_2}$} & \qw & \qw & \qw & \targ{} & \gate{R_\mathrm{z}(-\gamma)} & \targ{} & \qw & \qw & \qw & \qw & \qw & \qw & \qw \hdots \hdots\\
\lstick{$\ket{r_3}$} & \qw & \qw & \qw & \qw & \qw & \qw & \targ{} & \gate{R_\mathrm{z}(-\gamma)} & \targ{} & \qw & \qw & \qw & \qw \hdots \hdots\\
\lstick{$\ket{r_4}$} & \qw & \qw & \qw & \qw & \qw & \qw & \qw & \qw & \qw & \targ{} & \gate{R_\mathrm{z}(-\gamma)} & \targ{} & \qw \hdots \hdots\\
\lstick{$\ket{r_5}$} & \qw & \qw & \qw & \qw & \qw & \qw & \qw & \qw & \qw & \qw & \qw & \qw & \qw \hdots \hdots\\
\lstick{$\ket{r_6}$} & \qw & \qw & \qw & \qw & \qw & \qw & \qw & \qw & \qw & \qw & \qw & \qw & \qw \hdots \hdots\\
\lstick{$\ket{x_1}$} & \hdots \hdots & \qw & \qw & \qw & \qw & \qw & \qw & \qw &\\
\lstick{$\ket{x_2}$} & \hdots \hdots & \qw & \qw & \qw & \qw & \qw & \qw & \qw &\\
\lstick{$\ket{x_3}$} & \hdots \hdots & \qw & \qw & \qw & \qw & \qw & \qw & \qw &\\
\lstick{$\ket{x_4}$} & \hdots \hdots & \qw & \qw & \qw & \qw & \qw & \qw & \qw &\\
\lstick{$\ket{x_5}$} & \hdots \hdots & \ctrl{6} & \qw & \ctrl{6} & \qw & \qw & \qw & \qw &\\
\lstick{$\ket{x_6}$} & \hdots \hdots & \qw & \qw & \qw & \ctrl{6} & \qw & \ctrl{6} & \qw &\\
\lstick{$\ket{r_1}$} & \hdots \hdots & \qw & \qw & \qw & \qw & \qw & \qw & \qw &\\
\lstick{$\ket{r_2}$} & \hdots \hdots & \qw & \qw & \qw & \qw & \qw & \qw & \qw &\\
\lstick{$\ket{r_3}$} & \hdots \hdots & \qw & \qw & \qw & \qw & \qw & \qw & \qw &\\
\lstick{$\ket{r_4}$} & \hdots \hdots & \qw & \qw & \qw & \qw & \qw & \qw & \qw &\\
\lstick{$\ket{r_5}$} & \hdots \hdots & \targ{} & \gate{R_\mathrm{z}(-\gamma)} & \targ{} & \qw & \qw & \qw & \qw &\\
\lstick{$\ket{r_6}$} & \hdots \hdots & \qw & \qw & \qw & \targ{} & \gate{R_\mathrm{z}(-\gamma)} & \targ{} & \qw &\\
\end{quantikz}
    \caption{Cost unitary $U_f(\gamma)$ for the $[6,3,3]$ linear block code. Horizontal dots are used to represent the circuit continuation. $R_z(-\gamma)$ is the single qubit rotation gate about the $z$ axis, $R_z(-\gamma) = e^{i \frac{\gamma}{2} \mathrm{Z}}$.}
    \label{fig:appendix_cost}
\end{figure}

\newpage
\subsection{Success rate plots}
\label{ap:appendix4}
In Section \ref{sec:Performance Analysis of the UPO}, we perform an analysis on the performance of the UPO strategy.
Through numerical simulations, we argue that optimized solution is more achievable using UPO, rather than FPO.
Here, we demonstrate another set of results, which shows the promising performance of UPO.
We perform multiple simulations of the decoder and estimate an average success rate of the same.
From these results, it is very clear that UPO has better optimization capacity than FPO.
\begin{figure}[hbt!]
     \centering
     \begin{subfigure}[b]{0.49\textwidth}
         \centering
         \includegraphics[width = \textwidth]{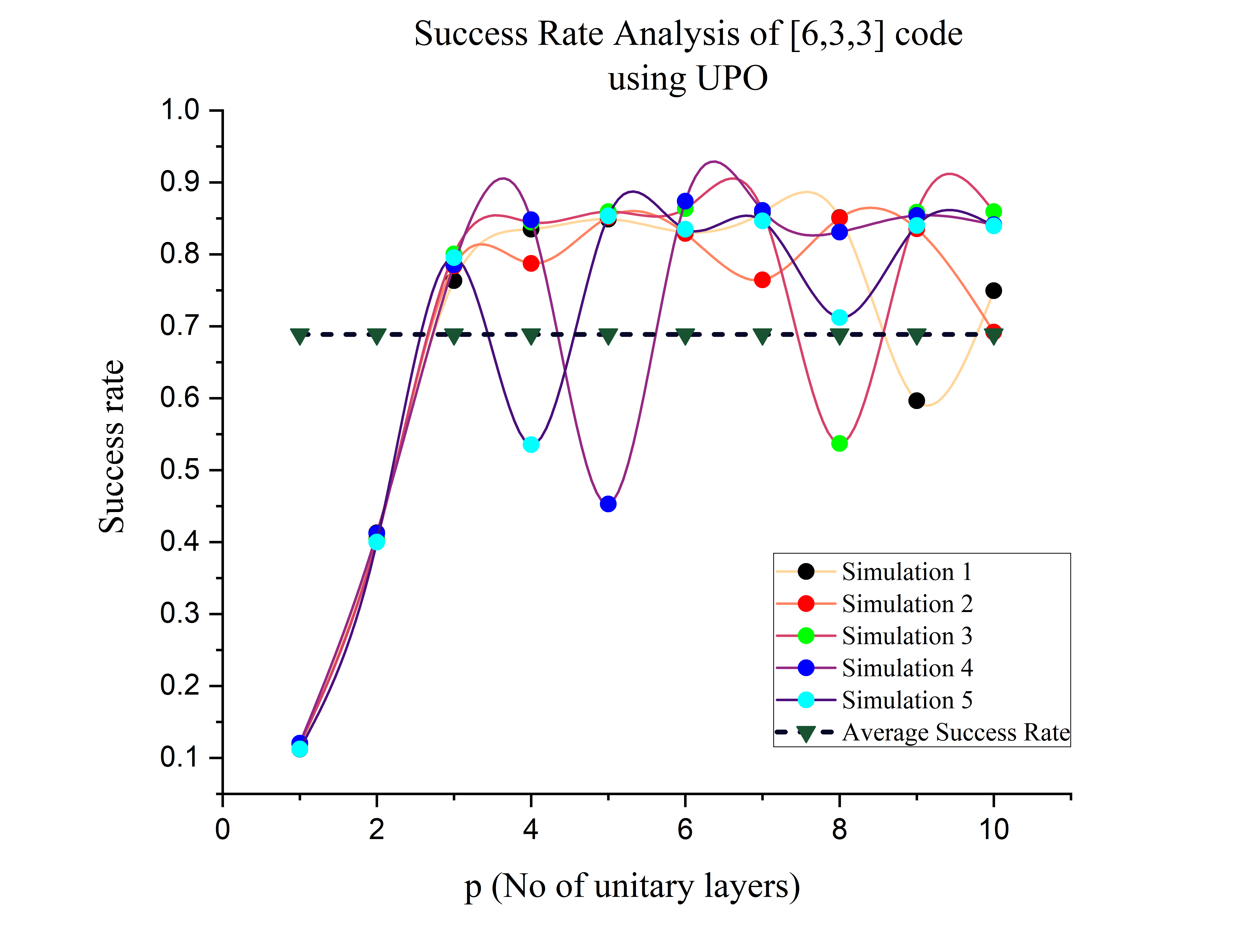}
         \caption{}
         \label{}
     \end{subfigure}
      \hfill
     \begin{subfigure}[b]{0.49\textwidth}
         \centering
         \includegraphics[width = \textwidth]{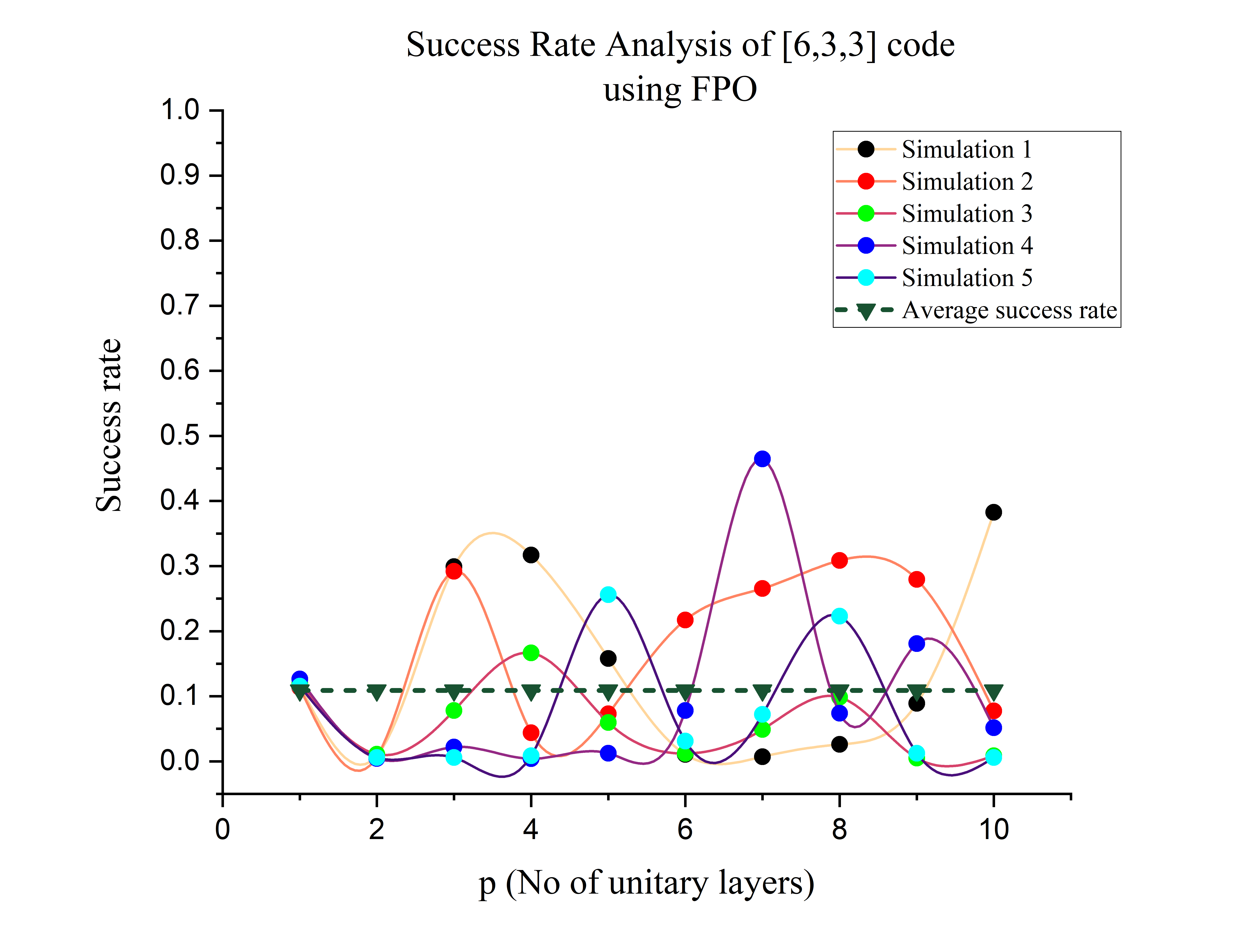}
         \caption{}
         \label{}
     \end{subfigure}
     \begin{subfigure}[b]{0.49\textwidth}
         \centering
         \includegraphics[width = \textwidth]{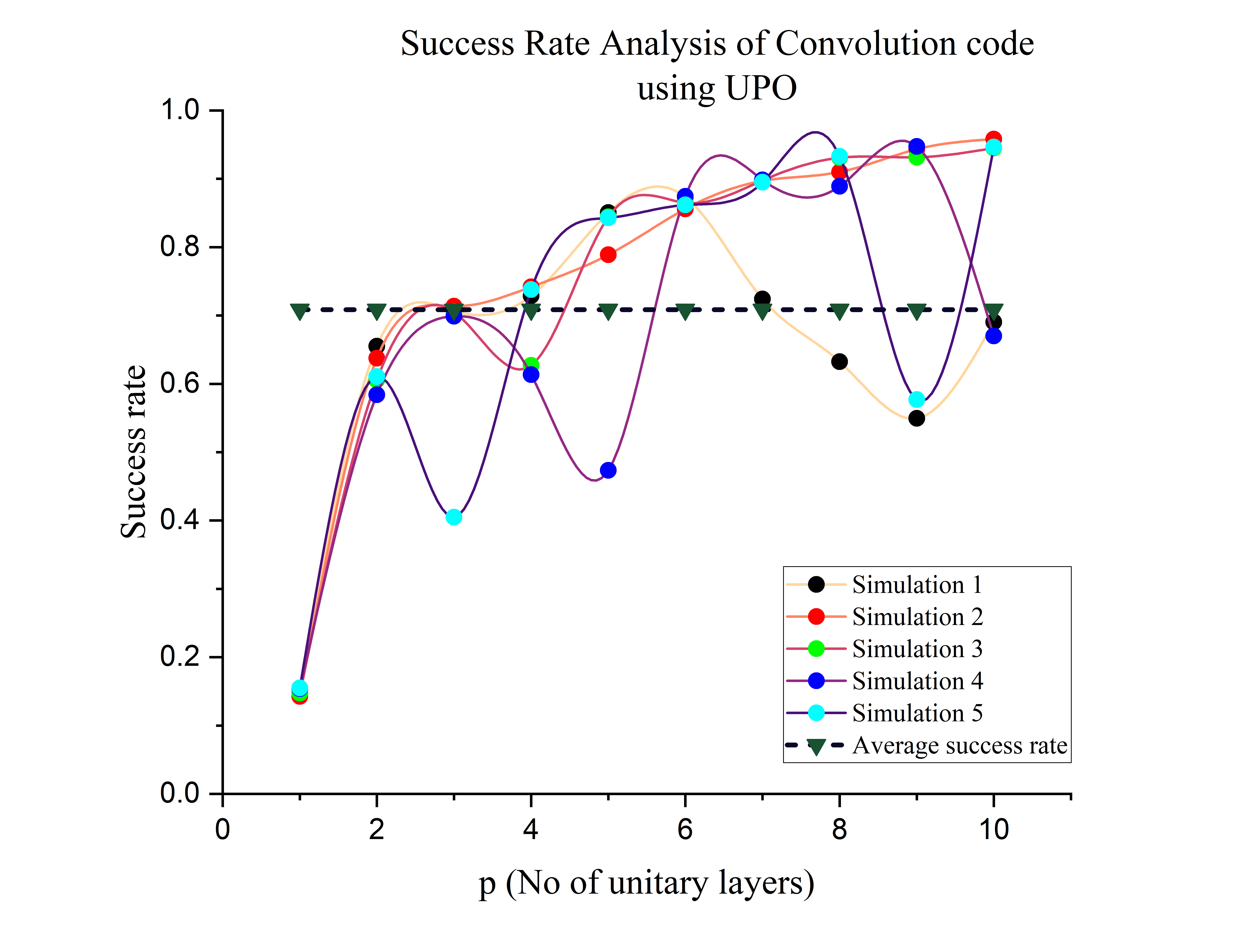}
         \caption{}
         \label{}
     \end{subfigure}
       \hfill
     \begin{subfigure}[b]{0.49\textwidth}
         \centering
         \includegraphics[width = \textwidth]{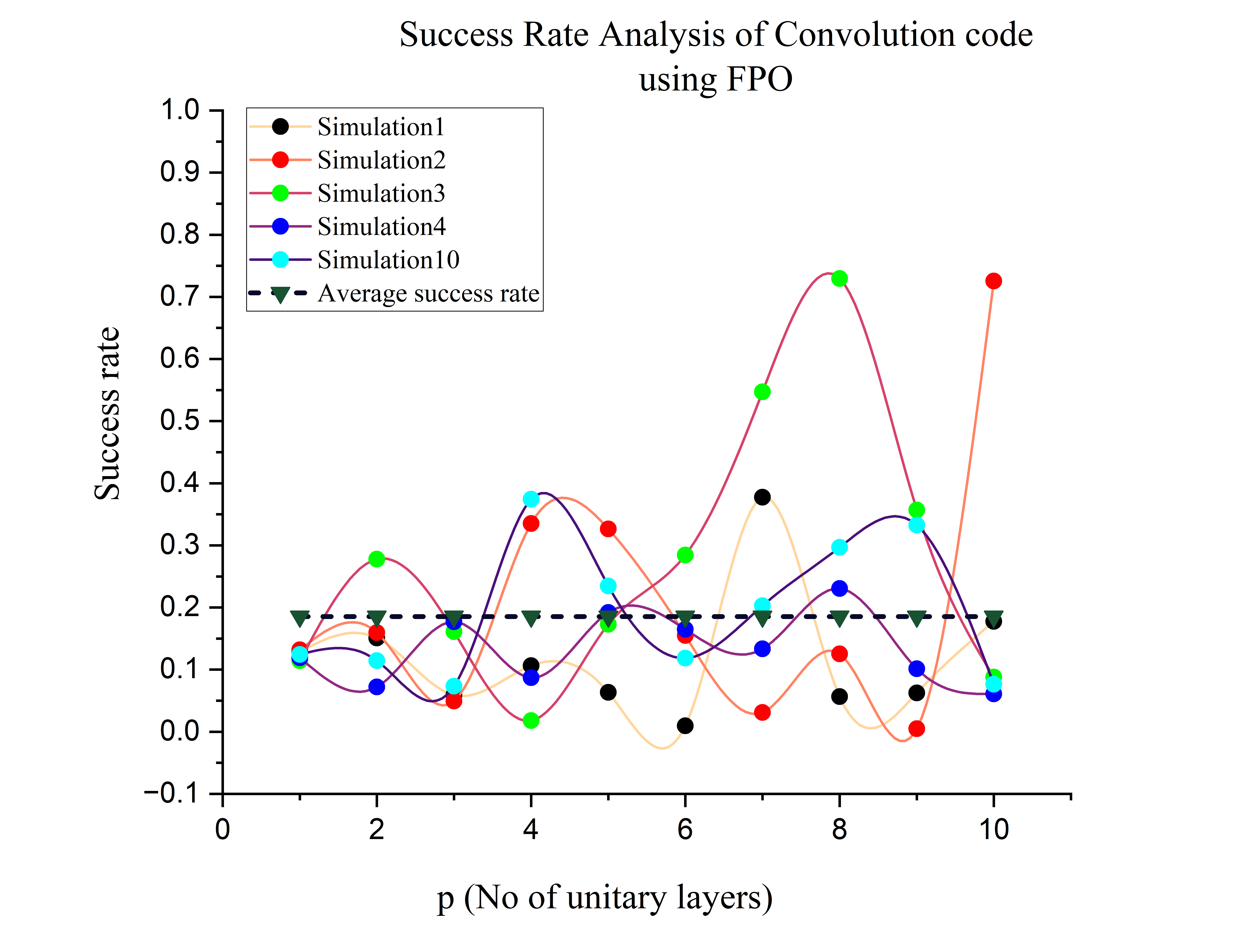}
         \caption{}
         \label{}
     \end{subfigure}
        \caption{Success rate plots.
        In (a) and (b) we show the success rate of the $[6,3,3]$ hybrid decoder with parameter optimization using UPO and FPO respectively.
        In (c) and (d) we show the same for a rate $\frac{1}{2}$ convolution code.
        For both the case, we observe that the average success rate is high for the case of uniformly optimized parameters.}
\end{figure}

\subsection{Application to Convolution Codes}
\label{ap:convolutioncodes}
In this appendix, we apply the hybrid Viterbi decoder from Section \ref{sec:Quantum Classical Hybrid Viterbi Decoder} and \ref{sec:UPO} to decode the classical convolution codes. 
Convolution codes does not have a block-like representation. 
The output codeword of a convolution code, depends on the current input and the memory elements of the encoder. 
The encoder here takes stream of input bits, produces the output codeword and updates the state of the encoder. 
The bits stored in the shift registers of the convolution encoder determine the state of the encoder. 
Shift registers are the memory elements, which are present in the circuit of every classical convolution code encoder.
The encoder takes a continuous stream of input bits and this memory elements store these information bits for the next sequence of input.
The number of past input bits stored, depends on the structure of the code and in general a $m$ bit shift register stores $m$ number of previous input bits.
Here, we use a rate $\frac{1}{2}$ code, with the number of message bits $k=1$ and the number of output codeword bits $n = 2$, at each time instant of the encoder. 
The number of shift registers present in the encoder is $m = 2$.
We let the encoder run for $5$ time instances. 
So, the codeword has a total of $10$ bits. 
At each time instant $t$, the input bit is defined as $u_t$. Then the state of the encoder we define as $S_t = (u_{t-1}\,\,u_{t-2})$,
where $u_{t-1}$ and $u_{t-2}$ are the input bits at time instants $t-1$ and $t-2$ respectively.
The trellis of this rate $\frac{1}{2}$ encoder is shown in Fig. \ref{fig:con_trellis}.
\begin{figure}
\centering
\begin{tikzpicture}[node distance={2.8 cm}, main/.style = {draw, circle,fill=cyan,thick}]

\node[main] (1){$\mathbf{00}$};

\node[main] (2)[right of=1]{$\mathbf{00}$};
\node[main] (3)[below of=2]{$\mathbf{01}$};
\node[main] (4)[below of=3]{$\mathbf{10}$};
\node[main] (5)[below of=4]{$\mathbf{11}$};

\node[main] (6)[right of=2]{$\mathbf{00}$};
\node[main] (7)[below of=6]{$\mathbf{01}$};
\node[main] (8)[below of=7]{$\mathbf{10}$};
\node[main] (9)[below of=8]{$\mathbf{11}$};

\node[main] (10)[right of=6]{$\mathbf{00}$};
\node[main] (11)[below of=10]{$\mathbf{01}$};
\node[main] (12)[below of=11]{$\mathbf{10}$};
\node[main] (13)[below of=12]{$\mathbf{11}$};

\node[main] (14)[right of=10]{$\mathbf{00}$};
\node[main] (15)[below of=14]{$\mathbf{01}$};
\node[main] (16)[below of=15]{$\mathbf{10}$};
\node[main] (17)[below of=16]{$\mathbf{11}$};

\node[main] (18)[right of=14]{$\mathbf{00}$};

\draw[->,very thick] (1) -- node[midway,above] {\textbf{00}} (2);
\draw[->,very thick] (1) -- node[midway,below,left = 0.15cm] {\textbf{11}} (4);
\draw[->,very thick] (2) -- node[midway,above] {\textbf{00}} (6);
\draw[->,very thick] (2) -- node[near start,above,right = 0.15cm] {\textbf{11}} (8);
\draw[->,very thick] (4) -- node[near start,midway,left = 0.15 cm] {\textbf{01}} (7);
\draw[->,very thick] (4) -- node[midway,left=0.25cm] {\textbf{10}} (9);
\draw[->,very thick] (6) -- node[midway,above] {\textbf{00}} (10);
\draw[->,very thick] (6) -- node[midway,right=0.2cm,above] {\textbf{11}} (12);
\draw[->,very thick] (7) -- node[near end,left=0.2cm] {\textbf{11}} (10);
\draw[->,very thick] (7) -- node[near start,right=0.1cm,above] {\textbf{00}} (12);
\draw[->,very thick] (8) -- node[near start,above,left=0.2cm] {\textbf{01}} (11);
\draw[->,very thick] (8) -- node[near end,right = 0.1cm,above] {\textbf{10}} (13);
\draw[->,very thick] (9) -- node[midway,above] {\textbf{01}} (13);
\draw[->,very thick] (9) -- node[near start,left=0.1cm] {\textbf{10}} (11);
\draw[->,very thick] (10) -- node[midway,above] {\textbf{00}} (14);
\draw[->,very thick] (11) -- node[midway,below,right=0.1 cm] {\textbf{11}} (14);
\draw[->,very thick] (12) -- node[midway,above,left=0.1cm] {\textbf{01}} (15);
\draw[->,very thick] (13) -- node[near start,below,right=0.1 cm] {\textbf{10}} (15);
\draw[->,very thick] (14) -- node[midway,above] {\textbf{00}} (18);
\draw[->,very thick] (15) -- node[midway,below,right=0.1cm] {\textbf{11}} (18);
\end{tikzpicture}
    \caption{Trellis of the rate $\frac{1}{2}$ convolution encoder \cite{lin}. Between initial and final state there are $5$ time instants. At each instant the input to the encoder is either $0$ or $1$. The branches are labeled according to the $2$ bit output at each time instant. There are two branches from each node due to two possible input at each instant.}
    \label{fig:con_trellis}
\end{figure}
\begin{figure}[ht!]
     \centering
     \begin{subfigure}[b]{0.48\textwidth}
         \centering
         \includegraphics[width = \textwidth]{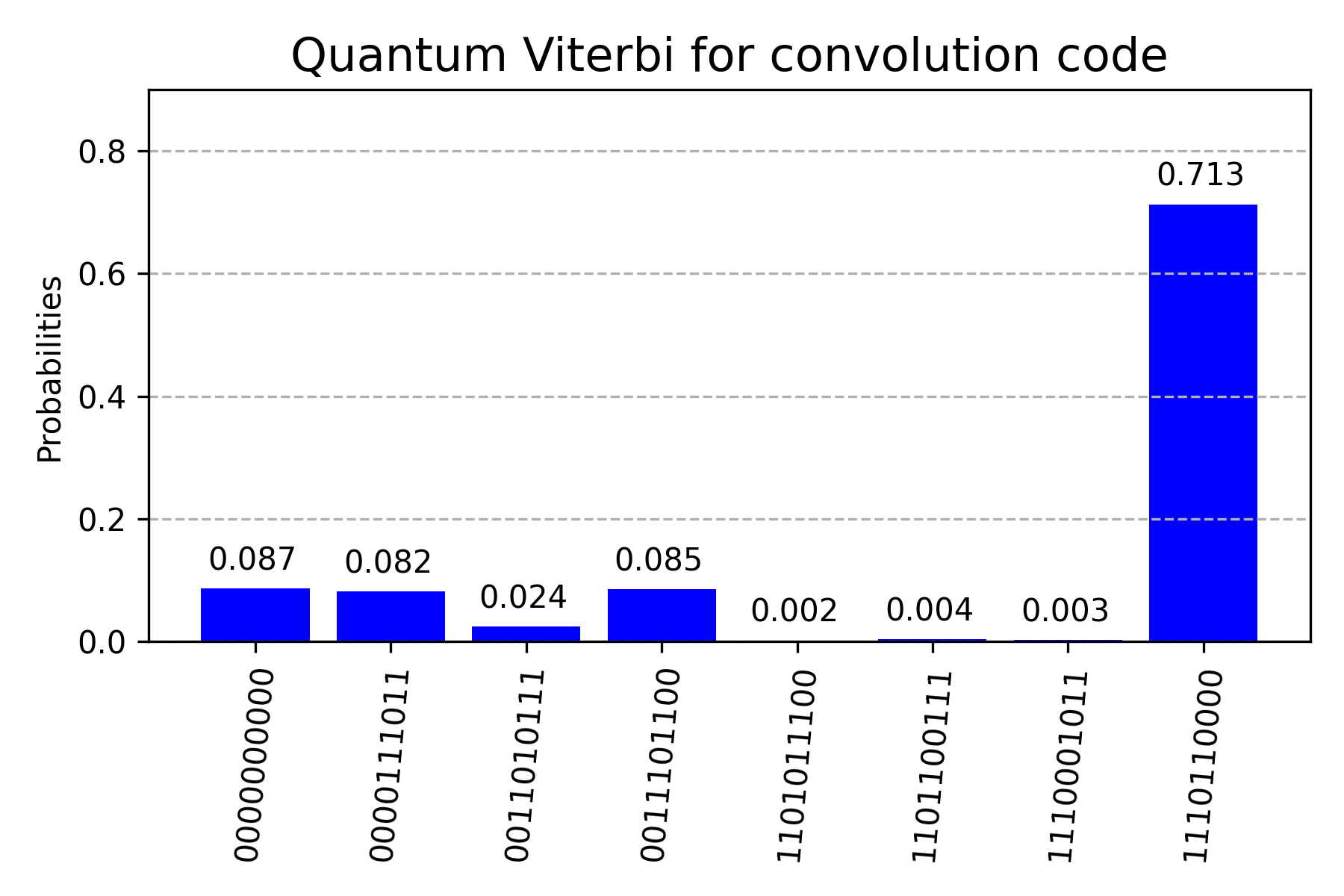}
         \caption{
         The received erroneous vector is $r = 0001110111$. Hamming distance between simulation output state and $r$ is $1$.}
         \label{fig:convl_0}
     \end{subfigure}
      \hfill
     \begin{subfigure}[b]{0.48\textwidth}
         \centering
         \includegraphics[width = \textwidth]{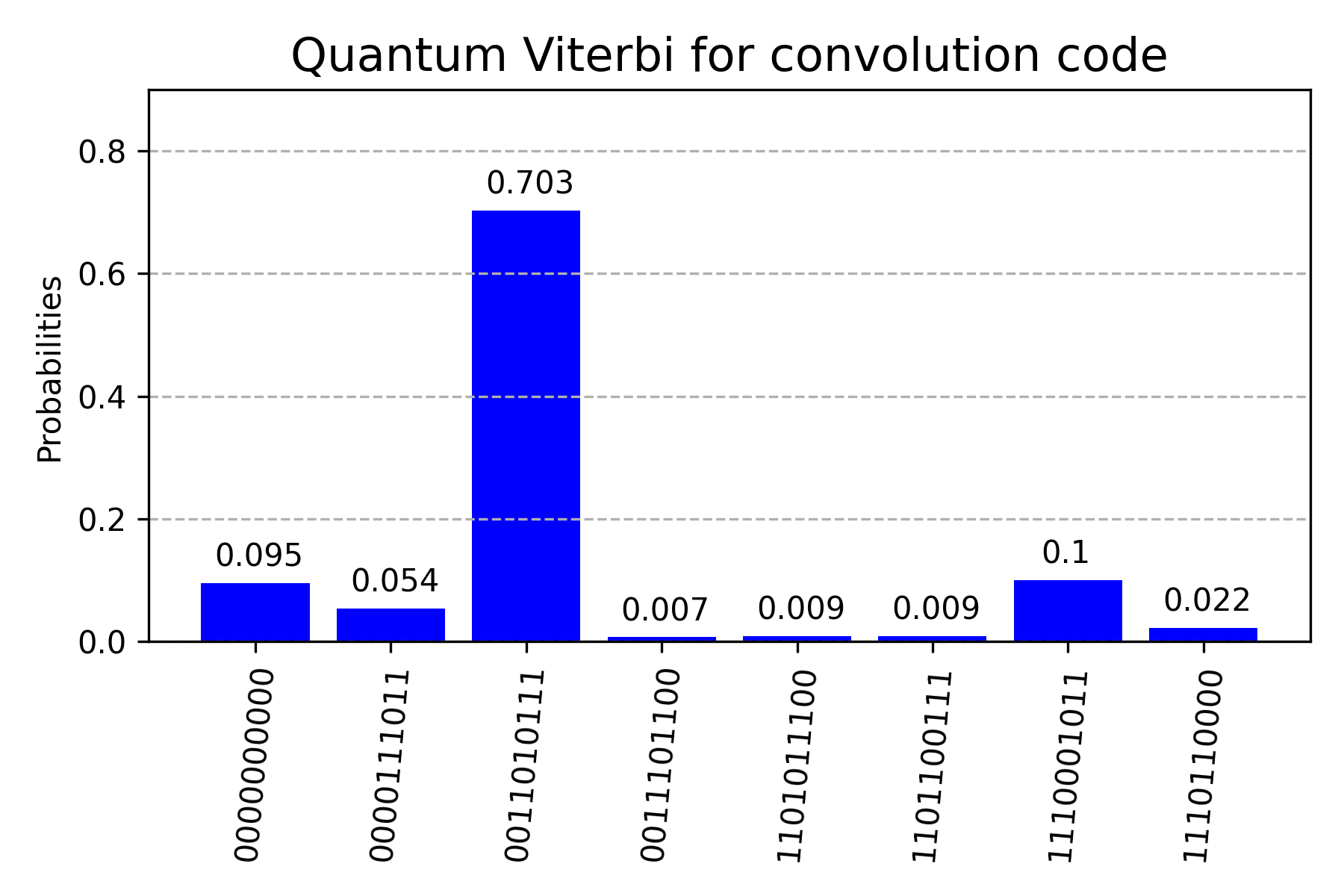}
         \caption{
         The received erroneous vector is $r =1110001100$. Hamming distance between simulation output state and $r$ is $1$.}
         \label{fig:convl_1}
     \end{subfigure}

        \caption{We show the results of the quantum Viterbi decoder for a rate $\frac{1}{2}$ convolution code, with $k = 1, n = 2, m = 2$, where $m$ is the number of memory elements. We label different states along the horizontal axis, with the convention that the most significant bit is at the bottom. The number of QAOA unitary layers applied is $p = 3$. 
        }
        \label{fig:convolution}
\end{figure}
The quantum Viterbi decoder uses UPO, mentioned in Algorithm \ref{alg:algorithm1}. 
The cost unitary is exactly the same as in Eq. \ref{eq:cost_unitary}. 
We note that $n$, the total number of codeword bits is $10$. 
The mixer Hamiltonian is constructed using the Theorem \ref{theorem_mixer}. 
The codespace $\mathcal{C}$ has $8$ members, corresponding to $8$ valid paths shown in the trellis Fig. \ref{fig:con_trellis}. 
The codespace is
\begin{align*}
    \mathcal{C} = \{&0000000000,0000110111,0011011100,0011101011,\\&1101110000,1101000111,1110101100,1110011011\}.
\end{align*}
The minimum Hamming weight of the code is $d = 5$. 
Using Theorem \ref{theorem_mixer}, we state the mixer Hamiltonian is
\begin{align}
     H_m = \mathrm{X}_1\mathrm{X}_2\mathrm{X}_4\mathrm{X}_5\mathrm{X}_6 + \mathrm{X}_3\mathrm{X}_4\mathrm{X}_6\mathrm{X}_7\mathrm{X}_8 + \mathrm{X}_5\mathrm{X}_6\mathrm{X}_8\mathrm{X}_9\mathrm{X}_{10}.
    \label{eq:mixer_con}
\end{align}
The corresponding mixer unitary is $U_m = e^{-i \beta H_m}$, where $H_m$ is used from Eq. \ref{eq:mixer_con}.
We show the results of simulations in Fig. \ref{fig:convolution} for different received vectors to determine the least paths on the trellis in Fig. 
\ref{fig:con_trellis}. 
The results of Fig. \ref{fig:convolution} show the feasibility and flexibility of the algorithm for different classical codes.
\end{document}